\documentclass{cup-pan}
\usepackage[utf8]{inputenc}
\usepackage{blindtext}

\title{Implementing partisan symmetry:\\
Problems and paradoxes}

\author[1]{Daryl DeFord}
\author[2]{Natasha Dhamankar}
\author[2]{Moon Duchin}
\author[3]{Varun Gupta}
\author[2]{Mackenzie McPike}
\author[2]{Gabe Schoenbach}
\author[2]{Ki Wan Sim}

\affil[1]{Washington State University}
\affil[2]{Tufts University/Voting Rights Data Institute}
\affil[3]{University of Pennsylvania}

\corrauthor{Moon Duchin}

\runningauthor{DeFord et al.} 

\usepackage{amsmath,graphicx,tikz,bbding}

\usepackage{amsthm}

\newcommand{\R}{\mathbb{R}}

\newcommand{\vbar}{\overline{v}}
\newcommand{\vmed}{v_{\rm med}}
\newcommand{\MM}{{\sf MM }}
\newcommand{\PG}{{\sf PG }}
\newcommand{\PB}{{\sf PB }}

\newcommand{\V}{{\sf v}}
\newcommand{\J}{{\sf j}}

\newcommand{\D}{ \deltaup}
\newcommand{\discrep}{\mathop{\rm discrep}}
\newcommand{\ceil}[1]{\lceil #1 \rceil}

\theoremstyle{plain}
\newtheorem{theorem}{Theorem}
\newtheorem{defn}[theorem]{Definition}
\newtheorem{example}[theorem]{Example}

\begin{document}
\maketitle

\begin{abstract}
We consider the measures of partisan symmetry proposed for practical use in the political science literature, as clarified and developed in \cite{KKR}.
Elementary mathematical manipulation shows the symmetry metrics to have surprising properties that call their meaningfulness into question.
To accompany the general analysis, we study measures of partisan symmetry with respect to recent voting patterns in Utah, Texas, and North Carolina, flagging problems in each case.  Taken together, these observations should raise major concerns about the available techniques for quantitative scores of partisan symmetry---including the mean-median score, the partisan bias score, and the more general "partisan symmetry standard"---as the decennial redistricting begins.

\end{abstract}

\section{Introduction}
In the political science literature, there is a long legacy of work on partisan gerrymandering, or the act of drawing political boundary lines to secure extra seats for the party that controls the process.
One of the questions attracting the most attention has been to measure the degree of partisan advantage
secured by a particular choice of redistricting lines.  But to counteract partisan gaming requires
a baseline notion of partisan fairness---extra seats compared to what baseline?---that has proved elusive.
The family of fairness metrics with perhaps the longest pedigree is called {\em partisan symmetry}
scores \citep{Grofman-1983,King-Browning,KGGK-amicus,Grofman-King,King-amicus,Grofman-amicus}, 
which got a conceptual and empirical overview and a timely renewed endorsement in \citep{KKR}.
The partisan symmetry standard is premised on 
the intuitively appealing fairness notion that 
the share of representation awarded to one party with its share of the vote {\em should} also have been secured by the other party, had the vote shares been exchanged. For instance, if  Republicans achieve 40\% of the seats with 30\% of the vote, then it would be deemed fair for the Democrats to also achieve 40\% of the seats with 30\% of the vote.

At the heart of the symmetry ideal is a commitment
to the principle that half of the votes should secure
half of the seats.
There are several metrics in the symmetry family that derive 
their logic from this core axiom.
The {\em mean-median metric} is vote-denominated: it produces a signed number that is often described as measuring how far short of half of the votes a party can fall while still securing half the seats.
A similar metric, {\em partisan bias}, is seat-denominated. Given the same input, it is said to measure how much more than 
half of the seats will be secured
with half of the votes.
The ideal value of both of these scores is zero. 
These are two in a large family of  partisan metrics that can be described in terms of geometric symmetry of the "seats-votes curve."

The focus in the current work is to show that there
are serious obstructions to the practical implementation of symmetry standards.
This is of pressing current interest at the time of writing: we are in a period of intense public focus on redistricting reform and on the cusp of a new round of redistricting.  In 2018 alone, voter referenda led four states to pass constitutional amendments  (CO, MI, MO, OH), and another to write reforms
into state law (UT) in anticipation of 2021 redistricting.\footnote{At the time of writing, those measures stand in CO, MI, and OH. In Missouri, voters overwrote and partly erased the reform in the subsequent vote in 2020. In Utah, the state legislature repealed parts of the reform.}  In Utah, partisan symmetry was adopted as a criterion to be considered
by the new independent redistricting commission before plans would be approved.\footnote{From the voter initiative: "The Legislature and the Commission shall use judicial standards and the best available data and scientific and statistical methods, including measures of partisan symmetry, to assess whether a proposed redistricting plan abides by and conforms to the redistricting standards" that bar party favoritism.  The legislature later blocked this from taking effect.
Cf. Utah Code, Chapter 7 Title 20A, Chapter 19 Part 1, Para. 103, \url{https://le.utah.gov/xcode/Title20A/Chapter19/C20A-19\_2018110620181201.pdf}.}  We sound a note of caution here, showing that
the versions of these scores that are realistically useable
are eminently gameable by partisan actors and 
do not have reliable interpretations.  To be precise:  in each state we studied, the most extreme partisan outcomes for at least one political party are still achievable with a clean bill of health from the full suite of partisan symmetry scores.  Furthermore, the signed scores (like mean-median $\MM$ and partisan bias $\PB$) make frequent sign errors in terms of partisan advantage.

Utah itself gives strong evidence of the interpretation 
problems:  with respect to recent voting patterns, a
good symmetry score can only be achieved by a plan that secures a Republican congressional sweep; 
what's more, the popular symmetry scores described above flag all possible plans with any Democratic
representation as major {\em Republican} gerrymanders.

We argue that there is at present no workable framework to make good on the idea of partisan symmetry.  
A manageable symmetry standard requires a swing assumption of some kind because its logic is built on voting counterfactuals (namely, table-turning).  But this puts the most tenuous election modeling task, vote prediction, at the center of the methodology. 
And the symmetry framing requires that the metrics be insensitive to the fundamentally spatial nature of redistricting: there are many ways, not one way, for a vote pattern to shift by a given amount.\footnote{This is also the problem with  axioms for partisan fairness metrics such as McGhee’s "efficiency principle," which assumes that more votes will lead to more seats, absent manipulation.  In reality, the spatial distribution of a party's votes, and not just the number, controls the possibilities for representation.}\footnote{In the end, our conclusion is that methods centered on varying districting lines rather than varying votes should be preferred in the study of gerrymandering. This recommendation is consistent with a practice of perturbing recent vote patterns to get a view of the robustness of a gerrymander.  A robustness check or a trend projection can be a valuable element of a redistricting analysis while playing a secondary rather than primary role.}

The main contribution of the paper is a Characterization Theorem for partisan symmetry under linear swing that clarifies what is actually measured by the leading symmetry scores.  
Applying this characterization, we demonstrate that realistic conditions can easily lead to "paradoxes" where one party has an extreme advantage in seats but the other party is flagged as the beneficiary of the gerrymander.  We then use recent electoral data from three states to demonstrate the ease of gaming symmetry scores and the prevalence of these paradoxical labelings.  

\subsection{Literature review}
\subsubsection{Building the seats-votes curve with 
available data}
We consider an election in a state with $k$ districts and two major parties, Party A and Party B. 
A standard construction in political science is the "seats--votes curve," 
a plot representing the relationship of the vote share for Party A to the seat share for the same party.  Observed outcomes generate single points
 in $V$-$S$ space---for instance, $(.3,.4)$ represents an election where Party A got 30\% of the votes and 40\% of the seats---but various methods have been used to 
extend from a scatterplot to a curve, such as by  fitting a curve from a given class (linear or cubic, for instance) to observed data points.  
We will 
focus on the construction of seats-votes curves that 
is emphasized in \cite{KKR}:  linear uniform partisan swing.  
Beginning with a single 
data point derived from a districting plan and a vote pattern,
the vote share is varied by a uniform shift, so that the district vote shares
$(v_1,\dots,v_k)$ will shift to 
$(v_1+t,\dots,v_k+t)$.  
This generates a step function 
spanning from $(0,0)$ to $(1,1)$ in the $V$-$S$ plot, with a jump in seat share each time a district is pushed past 50\% vote share for Party A.  (See Figures~\ref{fig:sv-curve}-\ref{fig:oregon16} below for examples.)

Linear uniform partisan swing is the leading 
method proposed for use in implementation.
Katz--King--Rosenblatt explicitly 
make it Assumption 3 in their symmetry survey---and use it throughout the paper---noting that the curve-fitting alternative 
is more suited "for academic study... than for practical use"
in evaluation of plans. Katz et al. also mention Assumption 4, a stochastic generalization of uniform swing, as a preferred alternative to linear swing "whenever it makes a difference," and cite \cite{Brunell,Jackman} as examples that employ that model.  This would add many additional modeling decisions, so would be difficult to carry out authoritatively in a practical setting. 
Nevertheless, we will identify some differences in shifting to a stochastic UPS approach in notes below. In short, neither including low levels of noise nor employing a vote vector obtained by regression from several elections will compromise the main findings here.

In any event, it is simple uniform partisan swing that is prevalent in practical applications.  Grofman noted in 1983 that linear swing is preferred in practice to more 
sophisticated models \citep[n.14]{Grofman-1983}; it is touted as the standard technique in  \citep{Garand}; and 
it has been invoked as recently as 2019 in expert reports and testimony  \citep{Mattingly}.
Finally, Missouri voters actually wrote linear uniform partisan swing into their state constitution in 2018, requiring that an election index prepared by the state demographer  be subjected to a swing of $t=-.05,-.04,\dots,+.05$ to test a plan's responsiveness.\footnote{"[T]he electoral performance index shall be used to simulate elections in which the hypothetical statewide vote shifts by one percent, two percent, three percent, four percent, and five percent in favor of each party. The vote in each individual district shall be assumed to shift by the same amount as the statewide vote." Cf. Missouri State Constitution, Article III Section 3.}
Because the present analysis is designed to address the prospects for implementation, we therefore focus on the linear swing model.

\subsubsection{Deriving  symmetry scores from the seats-votes curve}
Given a seats-votes curve, many symmetry scores have been proposed; here, we  focus on the 
mean-median score $\MM$, the partisan bias score $\PB$, and the partisan Gini score $\PG$, which have all been considered for at least 35 years.  (Definitions are 
found in the next section.)
Grofman's  1983 survey paper \citep{Grofman-1983} lays out eight possible scores of asymmetry once a seats-votes curve has been set.  
His Measure 3 is vote-denominated bias, which would equal $\MM$ under linear uniform partisan swing; similarly
his Measure 4 corresponds to $\PB$, and 
Measure 7 introduces $\PG$.

Because the partisan Gini is 
defined as the area between the 
seats-votes curve and its reflection over the center (seen in the shaded regions in Figs~\ref{fig:sv-curve}-\ref{fig:oregon16}),
it is easily seen to "control" all the other 
possible symmetry scores:
when $\PG=0$, its ideal value, all partisan symmetry metrics also take their ideal values, including $\MM$ and $\PB$. This agrees with Katz--King--Rosenblatt \citep[Def 1]{KKR}, where the coincidence of the curve and its reflection, i.e., $\PG=0$, is called the "partisan symmetry standard."
In the current work, our Theorem~\ref{thm:equiv} gives 
precise necessary and sufficient conditions for the partisan symmetry standard to be satisfied.

The literature invoking $\MM$ and $\PB$ as measures of bias is too large to survey comprehensively. We note that the  interpretation of median-minus-mean as quantifying (signed) Party A advantage is fairly standard in the journal literature, such as: "The median is 53 and the mean is 55; thus, the bias runs two points against Party A (i.e., $53-55=-2$)"  \citep{MB}.  The connection to the seats-votes curve is also standard:
$\MM$ "essentially slices the S/V graph horizontally at the S = 50\% level and obtains the deviation of the vote from 50\%" \citep[p351]{Nagle}.\footnote{There is even more work centered on $\PB$ 
(notably \cite{King-Browning}), but it is more rarely used in conjunction with linear swing, since that assumption makes its values move in large jumps.}

We briefly note the impact of introducing stochasticity  on the analysis below, we note that modifying the seats--votes curve $\gamma$ by adding  noising terms with mean zero will change the precision of our findings but not the basic structure, replacing exact equalities in the Characterization Theorem with approximate equalities.  In particular, this does not impact the prevalence of "paradoxes," for two reasons.  First, when the curve $\gamma$ passes far from $(V,S)=(.5,.5)$, perturbations to $\gamma$ will not move it past the center point, which would be needed to change the sign of $\MM$ and $\PB$ (as explained below in \S\ref{sec:defs}).  Second, the standard mean-median score is simply calculated as the difference of the mean vote share by district and the median \citep[p173]{KKR}, and thus relies on no swing assumption at all!  Abandoning linear swing therefore does not fix the problems with the mean-median score, but only breaks its relationship to the seats--votes curve.

\subsubsection{Applying symmetry scores in practice}
The current work is designed to evaluate the techniques proposed by leading practitioners for practical use.
Political scientists and their collaborators have advanced these scores in amicus briefs spanning from 
{\em LULAC v. Perry} (2006)  \citep{KGGK-amicus} to {\em Whitford v. Gill}
(2018) \citep{King-amicus} to 
{\em Rucho v. Common Cause} (2019) \citep{Grofman-amicus}. 
The scores have been claimed to be
"reliable and difficult to manipulate" and authors 
have argued that while 
"Symmetry tests should deploy actual election outcomes"
(as we do here), they will nonetheless
"measure opportunity,"
i.e., give information about future performance \citep[p17,24]{King-amicus}.  That assertion is drawn from an {\em amicus} brief in the Whitford case explicitly 
proposing mean-median as a concrete choice of score for
this task.
As laid out earlier in the influential LULAC brief,
\begin{quote}“Models applying the symmetry standard are by their nature predictive, just as the legislators themselves are predicting the potential impact of the map on partisan representation. The symmetry standard and the resulting measures of partisan bias, however, do not require forecasts of a particular voting outcome. Rather, by examining all the relevant data and the potential seat divisions that would occur for particular vote divisions, it compares the potential scenarios and determines the partisan bias of a map, separating out other potentially confounding factors. Importantly, those drawing the map have access to the same data used to evaluate it, and no data is required other than what is in the public domain" \citep[p11]{KGGK-amicus}. 
\end{quote}
This paper takes up precisely this modeling task 
in the manner explicitly proposed by its authors.

\subsection{Premises and caveats}

\subsubsection{What is partisan gerrymandering?}
To assess the success of partisan symmetry metrics at their task of identifying partisan gerrymandering, we should be clear about first principles.  
First, we agree on the definition from \cite{KKR}:  "Partisan gerrymanderers use their knowledge of voter preferences and their ability to draw favorable redistricting plans to maximize their party's seat share."  That is, the express intent of a partisan gerrymander is to secure for their party as many seats as possible under the constraints of voter geography and the other rules of redistricting.

This means that a (successful) gerrymander in favor of Party A is a districting plan that obtains an extreme Party A seat share.  In the public perception, that will usually be assessed by comparing the seat share to the vote share, undergirded by an intuition that equates fairness with {\em proportionality}. But proportionality is not the neutral tendency of redistricting and in some cases it may be literally impossible to secure \citep{VRDI-MA}.
For the better part of a century, political scientists have investigated this neutral tendency by appealing to constructions like seats--votes curves and cube laws.  
This literature has been severely limited by its inattention (with a few notable exceptions) to spatial factors, i.e., to the geography of the vote distribution.\footnote{The absence of vote geography in the bulk of redistricting research is noted, for instance, by \cite{Johnston2002,Rodden,CalvoRodden}.}  
A powerful alternative has recently emerged through so-called {\em ensemble methods}: Markov chain algorithms (for example) can now build samples of alternative districting plans, holding a vote distribution fixed.  Though we make use of ensembles of alternative plans below, the crux of this paper does not require the reader to commit to this or any particular choice of non-gerrymandered baseline.  On any common view of the baseline, from proportionality to cube law to outlier analysis from an ensemble, the standard definition of partisan gerrymandering entails views like these:
\begin{itemize}
    \item Circa 2016, the voter preferences in North Carolina were fairly even between Democratic and Republican candidates for statewide office.  Both algorithmic techniques and human mapmakers can easily draw plans ranging from 7 to 10 districts with a Republican majority in the 13-member Congressional plan.  
    In this context, a successful Republican gerrymander would secure a 10R-3D outcome, or an even more extreme outcome if possible.
    \item In Utah, voting patterns in this period tend to favor Republican over Democratic candidates by a roughly 70-30 split.  This is tilted enough that a 4R-0D  Congressional plan is in some sense typical or expected, and need not be viewed as a gerrymander.  However, it would be an error to label a 3R-1D outcome as a Republican-favoring gerrymander.
\end{itemize}
We will treat these as premises in the treatment below.

\subsubsection{Ensembles, not estimators}
In this paper, the algorithmically generated plans are not offered
as a statistical experiment and come with no probabilistic claims, but mainly provide {\em existence proofs} to illustrate how readily gameable partisan symmetry
standards will be for those engaged in redistricting.\footnote{This is exactly the use of ensemble methods that is endorsed in \citep[p176]{KKR} as productive and compelling:  a demonstration of possibility and impossibility.}
The methods also produce many thousands of examples
of plans that are paradoxical in the sense developed in this paper, where signed partisan symmetry metrics identify the 
wrong party as the gerrymanderers.

The algorithm used here (described in \S\ref{sec:methods}) builds a sample of plans that are plausible by the lights of traditional districting principles:  they are population-balanced, contiguous, and relatively compact, using whole precincts as the building blocks. There are techniques to layer in other criteria in addition to these in a state-specific way to set up a more thorough outlier analysis, but that is not needed for this application.
(See for instance \cite{VA-criteria,Compet,VRA-ensembles}.)
Nothing here, or in the broader logic of ensemble analysis, assumes that line-drawers are random agents.

\subsubsection{The competitiveness caveat}

Writing after the {\em LULAC v. Perry} decision in 2007, Grofman and King offer this key caveat: "[W]e are not proposing to apply this methodology in every situation, but only in {\em potentially} competitive jurisdictions, where the consequences of gerrymandering might be especially onerous in thwarting the will of the majority” (\cite{Grofman-King}, their emphasis). In the following paragraph they suggest that "reasonably competitive” settings could be those where each party receives 40-60\% of vote share. This is a sizeable limitation on the scope of the symmetry approach:  only about half of states have a recent U.S. Senate voting pattern in this range, for example.\footnote{See for instance
\url{https://github.com/gerrymandr/party-tilt}.}
Two of the three cases presented below (North Carolina and Texas) are in this reasonably competitive zone; the third is Utah, where a partisan symmetry standard was recently enacted in law.  We could not find any record of political scientists speaking out against Utah's adoption of this measure while it was on the ballot in 2018.

The more recent \cite{KKR} mainly places its competitiveness caveat in the article supplement, though it is obliquely referenced in Assumption 2, which requires that there is a sufficiently large range of "possible values" for vote outcomes.
Even there, the caveat is hedged: "Although Assumption 2 is defined in terms of possible electoral outcomes, those that are exceedingly unlikely, such as Washington DC voting overwhelming[ly] Republican, do not violate this assumption but may generate model dependence in estimation."  On our reading, the authors do not rule out the use of partisan symmetry metrics even on states with an extreme partisan lean.  

In any case, the analysis presented here, which shows that the partisan symmetry standard devolves to a simple numerical test, extends to competitive as well as uncompetitive situations.

\section{A mathematical characterization
of the Partisan Symmetry Standard}\label{sec:defs}
We begin with definitions and notation needed to 
state  Theorem~\ref{thm:equiv}, which characterizes when 
$\PG=0$ (the Partisan Symmetry Standard from \cite{KKR}).
We describe the vote outcome in the election using an ordered tuple (i.e., a vector) whose coordinates record the Party A share of the two-party vote in each of the $k$ districts as follows: $\V=(v_1, \dots, v_k),$ where $0 \leq v_1 \le \ldots \le v_k \leq 1$.  
Let the mean district vote share for Party A be denoted $\vbar = \frac 1k \sum v_i$ and the  median district vote share, $\vmed$, be the median 
of the $\{v_i\}$, which equals
$v_{(k+1)/2}$ if $k$ is odd and 
$\frac 12 (v_{k/2}+v_{(k/2)+1})$ if
$k$ is even because of the convention that coordinates are in non-decreasing order.
We note that $\vbar$ is not 
necessarily the same as the statewide share for Party A except in the idealized scenario that the districts
have equal numbers of votes cast (i.e., equal turnout).

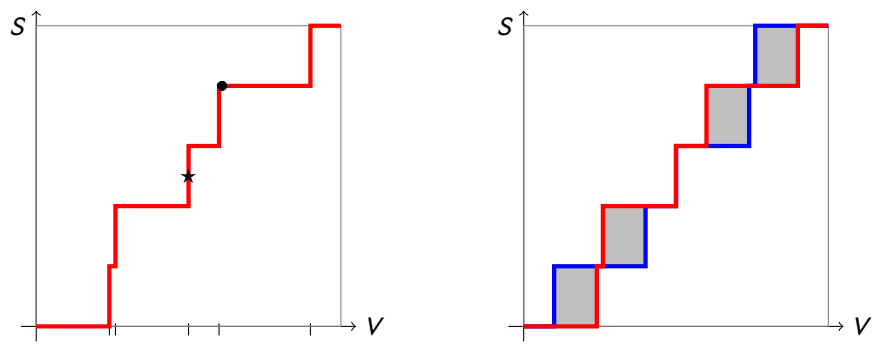
\begin{figure}[ht]
\centering
\begin{tikzpicture}[scale=4]
\def\ja{.24}
\def\jb{.26}
\def\jc{.5}
\def\jd{.6}
\def\je{.9}
\begin{scope}
\draw [->] (-.05,0)--(1.05,0);
\node at (1.05,0) [right] {$V$};
\draw [->] (0,-.05)--(0,1.05);
\node at (0,1) [left] {$S$};
\foreach \x in {\ja,\jb,\jc,\jd,\je}
{\draw (\x,-.03)--(\x,.01);}
\draw [gray,thin] (0,0) rectangle (1,1);
\draw [ultra thick,red] (0,0)--(\ja,0)--(\ja,1/5)--(\jb,1/5)
--(\jb,2/5)--(\jc,2/5)--(\jc,3/5)--(\jd,3/5)--(\jd,4/5)--
(\je,4/5)--(\je,1)--(1,1);
\filldraw (.61,4/5) circle (.015);
\node at (.5,.5) {$\star$};
\end{scope}
\begin{scope}[xshift=1.6cm]
\draw [->] (-.05,0)--(1.05,0);
\node at (1.05,0) [right] {$V$};
\draw [->] (0,-.05)--(0,1.05);
\node at (0,1) [left] {$S$};
\draw [gray,thin] (0,0) rectangle (1,1);
\filldraw [gray,fill opacity=.5,line width=0]
(0,0)--(\ja,0)--(\ja,1/5)--(\jb,1/5)
--(\jb,2/5)--(\jc,2/5)--(\jc,3/5)--(\jd,3/5)--(\jd,4/5)--
(\je,4/5)--(\je,1)--(1,1)--
(1-\ja,1)--(1-\ja,4/5)--(1-\jb,4/5)
--(1-\jb,3/5)--(1-\jc,3/5)--
(1-\jc,2/5)--(1-\jd,2/5)--(1-\jd,1/5)
--(1-\je,1/5)--(1-\je,0)--cycle;
\draw [ultra thick,blue,rotate=180,xshift=-1cm,yshift=-1cm] (0,0)--(\ja,0)--(\ja,1/5)--(\jb,1/5)
--(\jb,2/5)--(\jc,2/5)--(\jc,3/5)--(\jd,3/5)--(\jd,4/5)--
(\je,4/5)--(\je,1)--(1,1);
\draw [ultra thick,red] (0,0)--(\ja,0)--(\ja,1/5)--(\jb,1/5)
--(\jb,2/5)--(\jc,2/5)--(\jc,3/5)--(\jd,3/5)--(\jd,4/5)--
(\je,4/5)--(\je,1)--(1,1);
\end{scope}
\end{tikzpicture}
\caption{Red: The seats-votes curve 
$\gamma$ generated by  vote shares
 $\V=(.21,.51,.61,.85,.87)$ 
under uniform partisan swing.
This gives $\vbar=.61$ as the average vote share across the  districts.  The {\em jump points}, where an additional seat changes hands, are marked on the $V$ axis.
Blue: the reflection of $\gamma$
about the center point $\star$.
Since $\MM$ is the horizontal displacement from $\star$
to a point on $\gamma$,
this hypothetical election
has a perfect $\MM=0$ score, but it is not very symmetric overall, with $\PG=.112$, seen as the area of the shaded
region between $\gamma$ and its reflection. Because the step function jumps at $V=.5$, it is not clear how $\PB$ is defined in this case.}
\label{fig:sv-curve}
\end{figure}

The number
of districts in which 
Party A has more votes than 
Party B in an election with
vote shares $\V$ is the seat outcome,
$\#\{i:v_i>\frac 12\}$.
This induces a seats-votes function $\gamma=\gamma_\V$ defined as 
$\gamma(\vbar+t)=\#\{i: v_i+t>\frac 12\}/k$, which we can interpret as the share of 
districts won by Party A in 
the counterfactual that 
an amount $t$ was added to A's
observed vote share in every district.  
Varying $t$ to range over the one-parameter family of 
vote vectors 
$(v_1+t,\dots,v_k+t)$
is known as (linear)
{\em uniform partisan swing}.
The curve $\gamma$, treated
as a function $[0,1]\to [0,1]$, has been regarded 
as measuring how a fixed districting plan would behave
if the level of vote for Party A were to swing up or down.
Below, we will refer to the function and its
graph interchangeably, and we will
call it the {\em seats-votes curve}
associated to the vote share vector $\V$.
We begin with several scores based on $\V$ and $\gamma$.

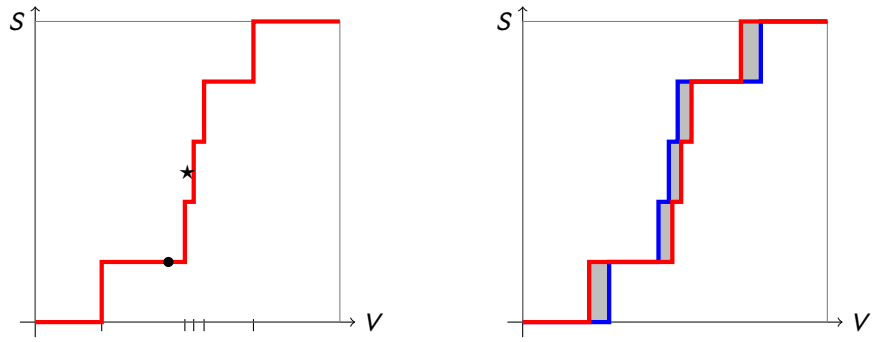
\begin{figure}[ht]
\centering
\begin{tikzpicture}[scale=4]
\def\ja{.2182}
\def\jb{.4912}
\def\jc{.5202}
\def\jd{.5542}
\def\je{.7162}
\begin{scope}
\draw [->] (-.05,0)--(1.05,0);
\node at (1.05,0) [right] {$V$};
\draw [->] (0,-.05)--(0,1.05);
\node at (0,1) [left] {$S$};
\foreach \x in {\ja,\jb,\jc,\jd,\je}
{\draw (\x,-.03)--(\x,.01);}
\draw [gray,thin] (0,0) rectangle (1,1);
\draw [ultra thick,red] (0,0)--(\ja,0)--(\ja,1/5)--(\jb,1/5)
--(\jb,2/5)--(\jc,2/5)--(\jc,3/5)--(\jd,3/5)--(\jd,4/5)--
(\je,4/5)--(\je,1)--(1,1);
\filldraw [black] (0.4372,1/5)
circle (.015);
\node at  (0.5,.5) {$\star$};

\end{scope}
\begin{scope}[xshift=1.6cm]
\draw [->] (-.05,0)--(1.05,0);
\node at (1.05,0) [right] {$V$};
\draw [->] (0,-.05)--(0,1.05);
\node at (0,1) [left] {$S$};
\draw [gray,thin] (0,0) rectangle (1,1);
\filldraw [gray,fill opacity=.5,line width=0]
(0,0)--(\ja,0)--(\ja,1/5)--(\jb,1/5)
--(\jb,2/5)--(\jc,2/5)--(\jc,3/5)--(\jd,3/5)--(\jd,4/5)--
(\je,4/5)--(\je,1)--(1,1)--
(1-\ja,1)--(1-\ja,4/5)--(1-\jb,4/5)
--(1-\jb,3/5)--(1-\jc,3/5)--
(1-\jc,2/5)--(1-\jd,2/5)--(1-\jd,1/5)
--(1-\je,1/5)--(1-\je,0)--cycle;
\draw [ultra thick,blue,rotate=180,xshift=-1cm,yshift=-1cm] (0,0)--(\ja,0)--(\ja,1/5)--(\jb,1/5)
--(\jb,2/5)--(\jc,2/5)--(\jc,3/5)--(\jd,3/5)--(\jd,4/5)--
(\je,4/5)--(\je,1)--(1,1);
\draw [ultra thick,red] (0,0)--(\ja,0)--(\ja,1/5)--(\jb,1/5)
--(\jb,2/5)--(\jc,2/5)--(\jc,3/5)--(\jd,3/5)--(\jd,4/5)--
(\je,4/5)--(\je,1)--(1,1);

\end{scope}
\end{tikzpicture}
\caption{The seats-votes curve 
$\gamma$ generated by the vote share
vector $\V=(.221,.383,.417,.446,.719)$,
which was the observed outcome 
in the 2016 Congressional races 
in Oregon from the Republican point of view.  
This gives
a mean of $\vbar=0.4372$, and earned
Republicans 1 seat out of 5.
The blue curve is the reflection
of $\gamma$ about the center,
so it shows seats at each vote share from the 
Democratic point of view.  This could be 
regarded as a situation with reasonably good symmetry, since the red and blue curves are close.
Its scores are $\PG=.05248$, $\MM=-.0202$, and $\PB=-.1$.
The sign of the latter two scores  is thought to indicate a Democratic advantage.}
\label{fig:oregon16}
\end{figure}

\newpage
\begin{defn}[Partisan symmetry scores]
The {\bf partisan Gini} score $\PG(\V)$ is the area between the 
seats-votes curve $\gamma_\V$ and its reflection over the center point 
$\star=(\frac 12, \frac 12)$.
$$\PG(\V)=\int_0^1 \bigl|\gamma(x)-\gamma(1-x)+1\bigr|\ dx.$$
The {\bf mean-median} score is $\MM(\V)=\vmed-\vbar$.
The {\bf partisan bias} score is
$\PB(\V)=\gamma(\frac 12)-\frac 12$.
\end{defn}

These scores can all be
related to the shape of the seats-votes curve $\gamma$ 
(see Figures~\ref{fig:sv-curve},\ref{fig:oregon16}).
Partisan Gini measures the failure of $\gamma$
to be symmetric about 
the center point $\star=(\frac 12,\frac 12)$,
in the sense that it is always non-negative, and it equals 
zero if and only
if $\gamma$ equals its reflection.
Mean-median score is the horizontal 
displacement from $\star$ to a point on $\gamma$,\footnote{To see this, plug in $t=1/2-\MM-\vbar=1/2-\vmed$ to 
deduce that $(1/2-\MM,1/2)$ is on $\gamma$---note that this connection from arithmetic to geometry is only exact when $\gamma$ has been constructed with linear swing, which is the standing assumption in this paper.} which is why it is votes-denominated (vote-share being the variable on the $x$-axis).
Similarly, partisan bias is the 
vertical displacement from $\star$ to a point
on $\gamma$, and is therefore seats-denominated.
(We note that $(\frac 12,\gamma(\frac 12))$ is a well-defined
point unless there is a jump precisely at $1/2$,
which occurs if some $v_i=\vbar$ on the nose---this
is shown in Figure~\ref{fig:sv-curve} but should not happen with real-world data.)
Below, we will focus on $\MM$ instead of $\PB$, but 
we note that $\MM>0\implies \PB\ge 0$ because of the geometric interpretation: if $\gamma$ passes to the 
left of $\star$ and is nondecreasing, then it must
pass through or above $\star$.

We can see that the  curve $\gamma$, and consequently the partisan Gini score, is exactly characterized by the points at which Party A has added enough vote share to secure the majority in an additional district. For the following analysis, it will be useful to characterize this curve in terms of the $\V$ data.

\begin{defn}[Gaps and jumps]
The {\bf gaps} in a vote share
vector $\V$ can be written in a 
gap vector
$$\D = (\delta_1, \delta_2, \ldots, \delta_{k-1})
    = (v_2 - v_1, \ \ v_3 - v_2, \ \ \ldots, \ \  v_k-v_{k-1}).$$
    
The {\bf jump points} for vote share
vector $\V$ are the values of $\vbar +t$
such that some $v_i+t=\frac 12$.
We have 
$$    t_1 := \frac{1}{2} - v_k \ , \qquad
    t_2 := \frac{1}{2} - v_{k-1}  \ ,
    \qquad
    \dots 
    \qquad 
    t_k := \frac{1}{2} - v_1$$
as the times corresponding to these jumps,
so we can denote the jumps as 
$j_i=\frac 12 + \vbar - v_{k+1-i}$,
and the jump vector as $\J=(j_1,\dots,j_k)$.
\end{defn}
By definition of $\gamma$,
these jump points $j_i$, marked in the figures, 
are exactly the $x$-axis locations (i.e., the $V$ values)
at which $\gamma$ jumps
from $(i-1)/k$ to $i/k$.\footnote{Warning to the reader:  if you try to draw your own 
examples to test some of these results, be aware that not just any step function can be generated by a vote vector.  The jump points must satisfy
$\sum j_i = k/2$, which follows directly from summing $j_i=\frac 12+\vbar-v_{k+1-i}$.}
 
With this notation, we can rephrase and relate the various partisan symmetry scores.  For instance, the center-most
rectangle(s) formed between $\gamma$ and its reflection have height $2\PB$ and width $2\MM$, which lets us relate the scores.
For small $k$, these relationships reduce to extremely simple 
expressions:  $\PG=\frac 43 |\MM|$ when $k=3$, and 
$\PG=2|\MM|$ when $k=4$.  (Proved in the Supplement.)
 
For any number of districts, we obtain a  clean characterization of 
precisely when vote shares by district  satisfy the  
Partisan Symmetry Standard \citep[Def 1]{KKR}.

\begin{theorem}[Partisan Symmetry Characterization]\label{thm:equiv}
    Given $k$ districts with vote shares $\V$, jump vector $\J$, and gap vector $\D$, the following are equivalent:
    \begin{align}
    \setcounter{equation}{0}
        &\PG(\V) = 0 \tag{Partisan Symmetry Standard}\\
        &j_i + j_{k+1-i} - 1 = 0 \qquad  \forall  i  \tag{jumps}\\
        &\frac 12 \left(v_i + v_{k+1-i}\right) = \vbar \qquad \forall  i  \tag{mean vote}\\
                &\frac 12 \left(v_i + v_{k+1-i}\right) = \vmed \qquad \forall  i  \tag{median vote} \\
        &\delta_i = \delta_{k-i} \qquad \forall  i \tag{gaps}
    \end{align}
\end{theorem} 
The proof is included in the Supplement.
Note also that the theorem statement makes it clear that $\PG=0\implies \MM=0$ by comparing the third equality to the fourth, which fits with the earlier observation that partisan Gini  "controls" the other scores.

Theorem~\ref{thm:equiv} asserts that the partisan symmetry standard under linear swing is nothing 
but the requirement that the vote shares by 
district are distributed on the number line 
in a symmetric way.\footnote{Consider replacing simple linear swing by the modeling approach frequently used in stochastic uniform partisan swing, such as in the JudgeIt software package.  In this case,  rather than using a single election’s vote-share-by-district, a linear swing is applied to a different vote vector that is obtained by regression from several elections.  In this case, the Characterization Theorem still holds, applied to this inferred vote vector instead of the observed vote vector.}
In particular, this tells you at a glance that 
an election with vote shares $(.37,.47,.57,.67)$ in its districts 
rates as perfectly partisan-symmetric, while one 
with vote shares $(.37,.47,.57,.60)$ falls short.

\begin{figure}[ht]
\begin{tikzpicture}[scale=6]
\begin{scope}
\draw (0,0)--(1,0); 
\draw (0,-.05)--(0,.05);
\draw (1,-.05)--(1,.05);
\draw (0,.2)--(1,.2); 
\draw (0,.15)--(0,.25);
\draw (1,.15)--(1,.25);
\foreach \x in {.37,.47,.57,.67}
{\filldraw (\x,.2) circle (0.01);}
\foreach \x in {.25,.285,.4,.48,.5,.52,.6,.715,.75}
{\filldraw (\x+.14,0) circle (0.01);
\node at (.2,.2) {\huge \CheckmarkBold};
\node at (.2,0) {\huge \CheckmarkBold};
\node at (.5,-.15) {\huge fair};
}
\end{scope}
\begin{scope}[xshift=1.2cm]
\draw (0,0)--(1,0); 
\draw (0,-.05)--(0,.05);
\draw (1,-.05)--(1,.05);
\draw (0,.2)--(1,.2); 
\draw (0,.15)--(0,.25);
\draw (1,.15)--(1,.25);
\foreach \x in {.37,.47,.57,.6}
{\filldraw (\x,.2) circle (0.01);}
\foreach \x in {.25,.32,.4,.48,.5,.58,.6,.715,.75}
{\filldraw (\x+.09,0) circle (0.01);
\node at (.2,.2) {\huge \XSolidBrush};
\node at (.2,0) {\huge \XSolidBrush};
\node at (.5,-.15) {\huge not fair};
}
\end{scope}
\end{tikzpicture}
\caption{Four election outcomes, shown as vote shares by district.  On the left-hand side, the $v_i$ are symmetric about their center, so all partisan symmetry scores are perfect.  On the right-hand side, non-symmetric outcomes. The partisan symmetry standard can be eyeballed by a glance at the vote shares in the districts.}
\label{fig:eyeball}
\end{figure}
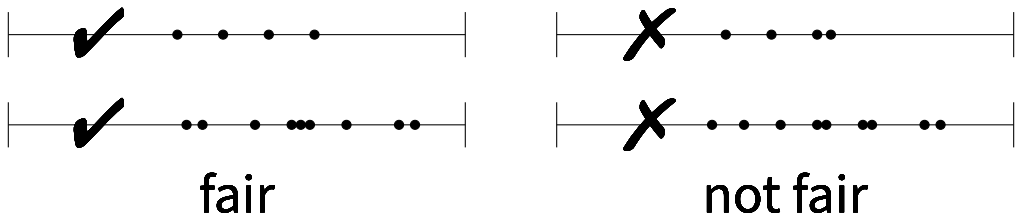

\section{Paradoxes with signed symmetry scores}

Recall that mean-median and partisan bias are {\em signed} scores that are supposed to identify which party has an advantage and by what amount. A positive score is said to indicate an advantage for Party A (the point-of-view party whose vote shares are reported in $\V$).
Let us say that a {\em paradox} occurs when the score
indicates an advantage for one party even though it has the fewest seats it can possibly earn with its vote share, say.\footnote{In practical settings, it may be hard to calculate the absolute minimum that is possible, so this may refer to an extremely low number of seats with respect to the known range of alternatives, as in the figures below.}
In other words, a paradox means that the score makes a sign error with respect to the standard definition of partisan gerrymandering. 

When there is an extremely skewed outcome (with a vote share for one party exceeding 75\%), we will show that paradoxes {\em always} occur, just as a matter of arithmetic.  But even for less skewed elections with a vote share between 62.5 and 75\% for the leading party---which frequently occur in 
practice!---mundane realities of political geography can force these sign paradoxes.  

To illustrate these observations, we will begin with the case
of $k=4$ districts, where the algebra is simpler.
The issues are not limited to small $k$, however:   in the empirical section we will find paradoxes
of this kind in $k=13$ and $k=36$ cases as well.

\begin{example}[Paradoxes forced by arithmetic] Suppose we have $k=4$
districts and an extremely skewed election in favor of 
Party $A$, achieving $75\% < \vbar < 87.5\%$.
With equal turnout, Party B can get at most one seat.
However, every vote vector $\V$ achieving
this outcome (one B seat) yields $\MM\ge \vbar-\frac 34> 0$.  In particular,
such districting plans  all have positive $\MM$ and $\PB$, 
paradoxically indicating an advantage for Party A in every case where Party B  gets representation.
\end{example}
    
The demonstration is simple arithmetic.  Since $\frac 12 (v_2+v_3) =\vmed$, we have 
$$\vbar=\frac 14 (v_1+v_2+v_3+v_4) = \frac{v_1+v_4}4 + \frac{v_2+v_3}4 \implies 
\vmed - \vbar = \vbar - \frac{v_1+v_4}2.$$    
Since $v_1\le \frac 12$ (for B to win a seat) and $v_4\le 1$, we get $\MM=\vmed-\vbar \ge \vbar-\frac 34$, as needed.
    
 A stronger statement can be made if one takes
political geography into account.
It was shown by Duchin--Gladkova--Henninger-Voss--Newman--Wheelen in a study of Massachusetts \citep{VRDI-MA} that, if the precincts are treated as atoms that are not to be split in redistricting, then several recent elections have the property that no choice of district lines can create even one district with  Republican vote share over $1/2$.  This is because Republican votes are distributed extremely uniformly across the precincts of the state.  While other states are not as uniform as Massachusetts, it is still true that there is some upper bound $Q$ on the vote share that is possible for each party in any district.  When this bound satisfies $Q<2\vbar-\frac 12$, even moderately skewed elections are forced to exhibit paradoxical symmetry scores.  As we will see below, having all $v_i<2\vbar-\frac 12$ ensures both that one seat is the best outcome for Party B and that the median vote share is greater than the mean.

\begin{example}[Paradoxes forced by geography]
Suppose we have $k=4$ districts and a skewed election 
in favor of Party A, with $62.5\%\le \vbar < 75\%$. 
Suppose the geography of the election has Party A 
support arranged uniformly enough that districts
can not exceed a share $Q$ of A votes, for some 
$Q<2\vbar-\frac 12$.  Then with equal turnout, Party B
can get at most one seat.  However, 
every vote vector $\V$ achieving this outcome (one B seat) has a positive
$\MM$ and $\PB$, paradoxically indicating an advantage
for Party A in every case where Party B gets representation.
\end{example}

\begin{proof}  First, it is easy to see that Party B can't achieve two seats:  in that case, we would have $v_1,v_2\le \frac 12$.  
Since we also have $v_3,v_4\le Q < 2\vbar-\frac 12$, we can average the $v_i$ to get the contradiction $\vbar<\vbar$.

To see that $\MM>0$, we write
$$\MM = \vmed - \vbar = \frac{v_2+v_3}2 - \frac{v_1+v_2+v_3+v_4}4 =  \frac{v_1+v_2+v_3+v_4}4 - \frac{v_1+v_4}2
=\vbar - \frac{v_1+v_4}2.
$$
Since $v_1<\frac 12$ and $v_4\le Q <2\vbar - \frac 12$, we have $\MM>\vbar - \frac{2Q+1}4>\vbar-\vbar=0$.  
\end{proof}

\section{Investigations with
observed vote data}

\subsection{Methods}\label{sec:methods}
In this section we illustrate the theoretical issues from above, using naturalistically 
observed election data together with a Markov chain technique that produces large ensembles of districting plans.\footnote{All data and code are public and freely available
for inspection and replication in GitHub \citep{gerrychain,PSymm} and Dataverse \citep{DVN/MVVLQC_2021}.}
In each case, we have run a recombination ("{\sf ReCom}") Markov chain for 100,000 steps---long enough to comfortably achieve
heuristic convergence benchmarks in all scores that 
we measured---while enforcing population balance, 
contiguity, and compactness.\footnote{The population 
balance imposed here is 1\% deviation from ideal 
district size.  Such plans are easily tuneable to 
1-person deviation by refinement at the block level without
significant impact to any other scores discussed here.
Contiguity is enforced by recording adjacency of precincts.
Compactness, at levels comparable to those observed
in human-made plans, is an automatic consequence of the 
spanning-tree-based recombination step.
For more information about the Markov 
chain used here, see \citep{ReCom}.}
Note that some Markov chain methods count every attempted move as a step, even though most proposals are rejected, so that each plan is counted with high multiplicity in the ensemble; in our setup, the proposal itself includes the criteria, and repeats are rare.  100,000 steps produces upwards of 99,600 distinct districting plans in each ensemble presented here.

We ran trials on multiple elections in our dataset and all results are available for comparison \citep{PSymm}.
Below, we  highlight the most recent available Senate race from a Presidential election year in each state, to match conditions across cases as closely as possible.
The data bottleneck is a precinct shapefile matching geography to voting patterns, which is surprisingly difficult to obtain.  
A database of shapefiles is available at \citep{MGGG-states}.  We use the statewide U.S. Senate election rather than an endogenous Congressional voting pattern because the latter is subject to uncontested races and variable incumbency effects.  For instance, Utah's 2016 Congressional race had all four seats contested, but a Republican vote in District 3 went for hard-right Jason Chaffetz, while on the other side of the invisible line to District 4, the vote went to Mia Love, a Black Republican who is outspoken on racial inequities. When the district lines are moved, it's not clear (for example) that a Love voter stays Republican.
The U.S. Senate race had a comparable number of total votes cast to the Congressional race (1,115,608 vs. 1,114,144) and offers a consistent choice of candidates around the state, making it better suited to methods that vary district lines.

\subsection{Utah and the "Utah Paradox"}
We begin with Utah, where the elections that were 
available in our dataset all come from 2016.\footnote{Out of GOV16, SEN16, and PRES16, none
gives an especially pure partisanship signal, because the Democratic candidates for Governor and Senate were quite weak, while the partisanship in the Presidential race was complicated by the presence of a very strong third-party candidate in Evan McMullin, giving that race an extremely different pattern.  The Governor's race shows similar results to the Senate, as the reader can verify in \citep{PSymm}.}
Utah has only four congressional districts and has a heavily skewed partisan preference, with a statewide Republican vote share of 71.55\% in the 2016 Senate race.\footnote{Recall that $\vbar$ is the average of the district vote shares, which will not generally equal the statewide share except under an equal-turnout assumption.}  Though it is far from clear that table-turning between the parties is conceivable in the near future,\footnote{The amount
of linear partisan swing needed to reverse the 
partisan advantage should be viewed as unreasonably
large under these conditions. With respect to SEN16, fully 199 out of 2123 precincts in the state have
Republican vote share that reaches zero under 
this amount of swing.  This is one of the reasons (though not the only reason) that this style of quantifying partisan advantage is poorly suited to Utah.}  Utahns nonetheless recently enacted partisan symmetry consideration into state law.

\begin{figure}
\centering
\begin{tikzpicture}
\node at (0,14) {Utah SEN16: Republican seats won across 100,000 {\sf ReCom} plans};
\node at (0,12)   {\includegraphics[height=1.4in]{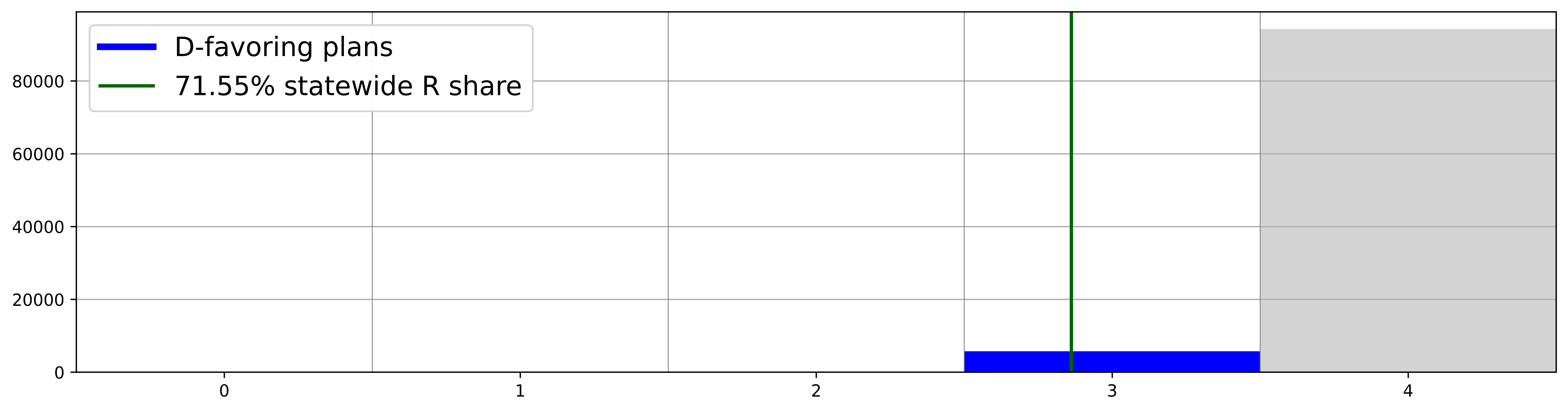}};
\node at (0,10) {Restriction to best $\MM$:
2609 plans with $|\MM|<.001$};
\node at (0,8)  {\includegraphics[height=1.4in]{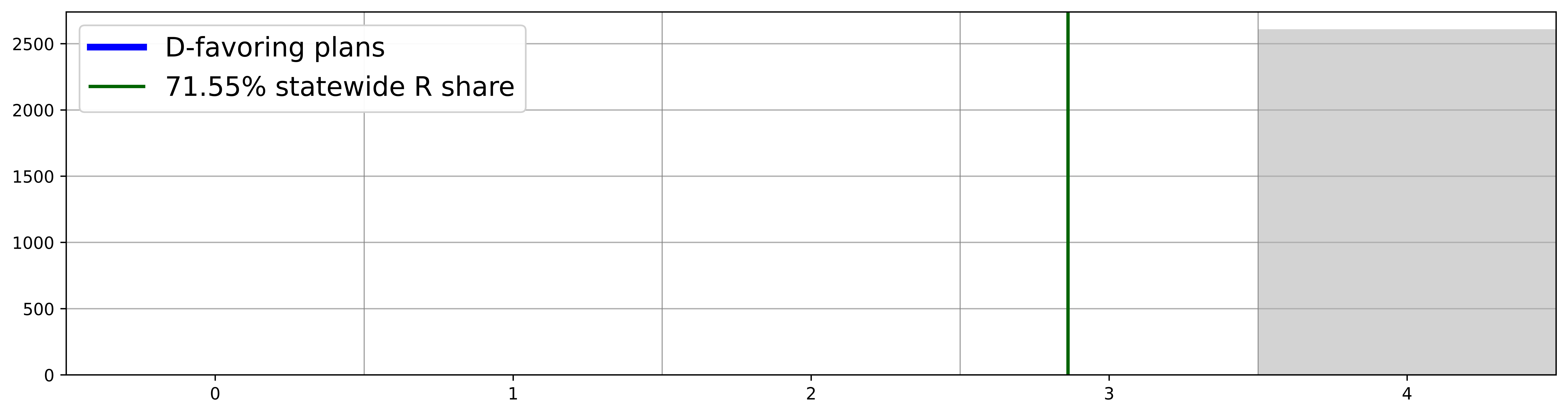}};
\node at (0,6) {Restriction to best $\PG$:
12,369 plans with $\PG<.01$};
\node at (0,4)  {\includegraphics[height=1.4in]{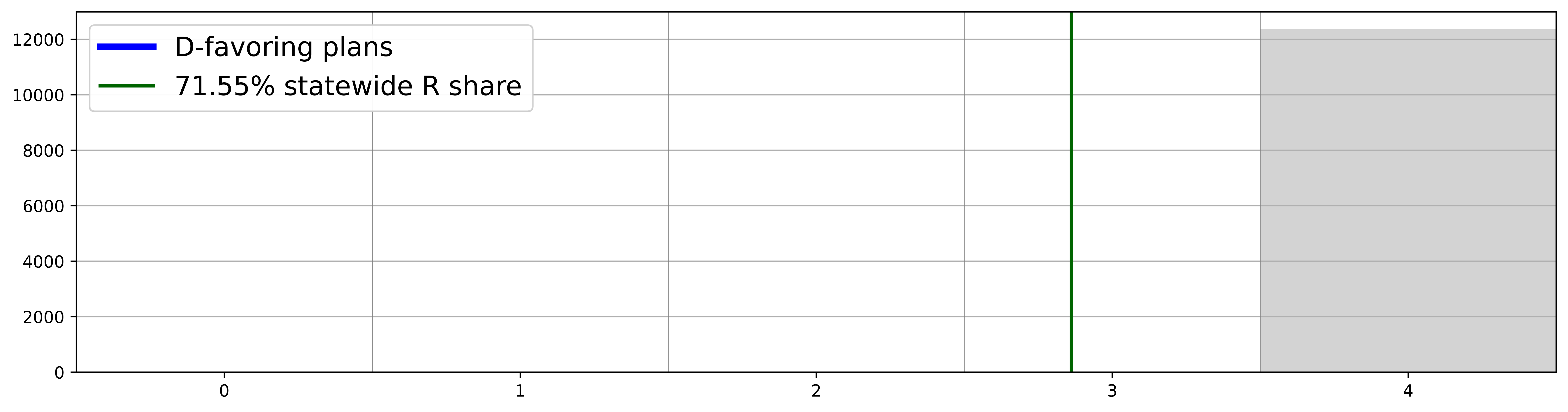}};

\node at (-3.5,0)  {\includegraphics[height=1.15in]{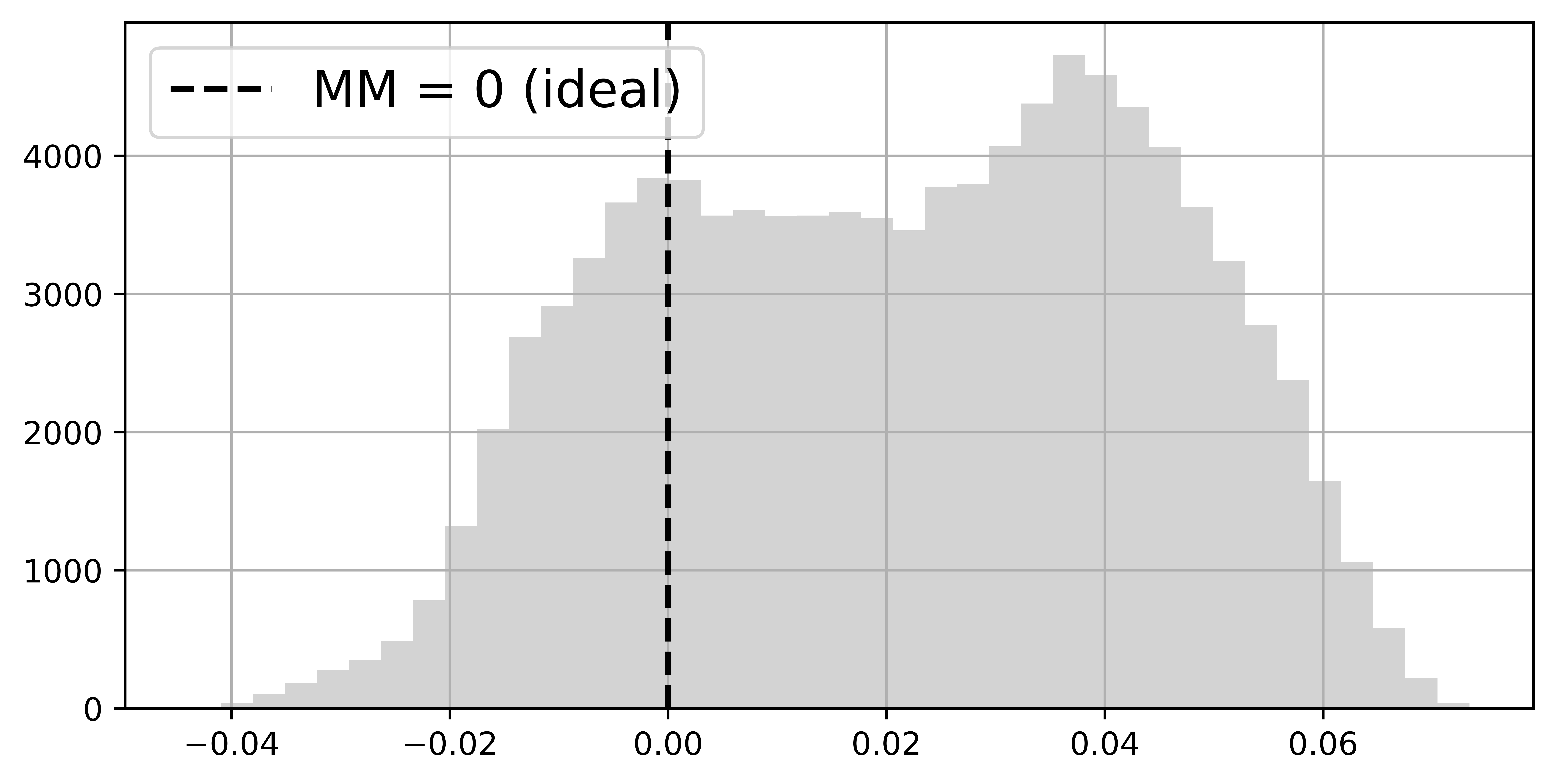}};
\node at (-3.5,1.7) {\small $\MM$ score: full ensemble};
\node at (3.5,0)  {\includegraphics[height=1.15in]{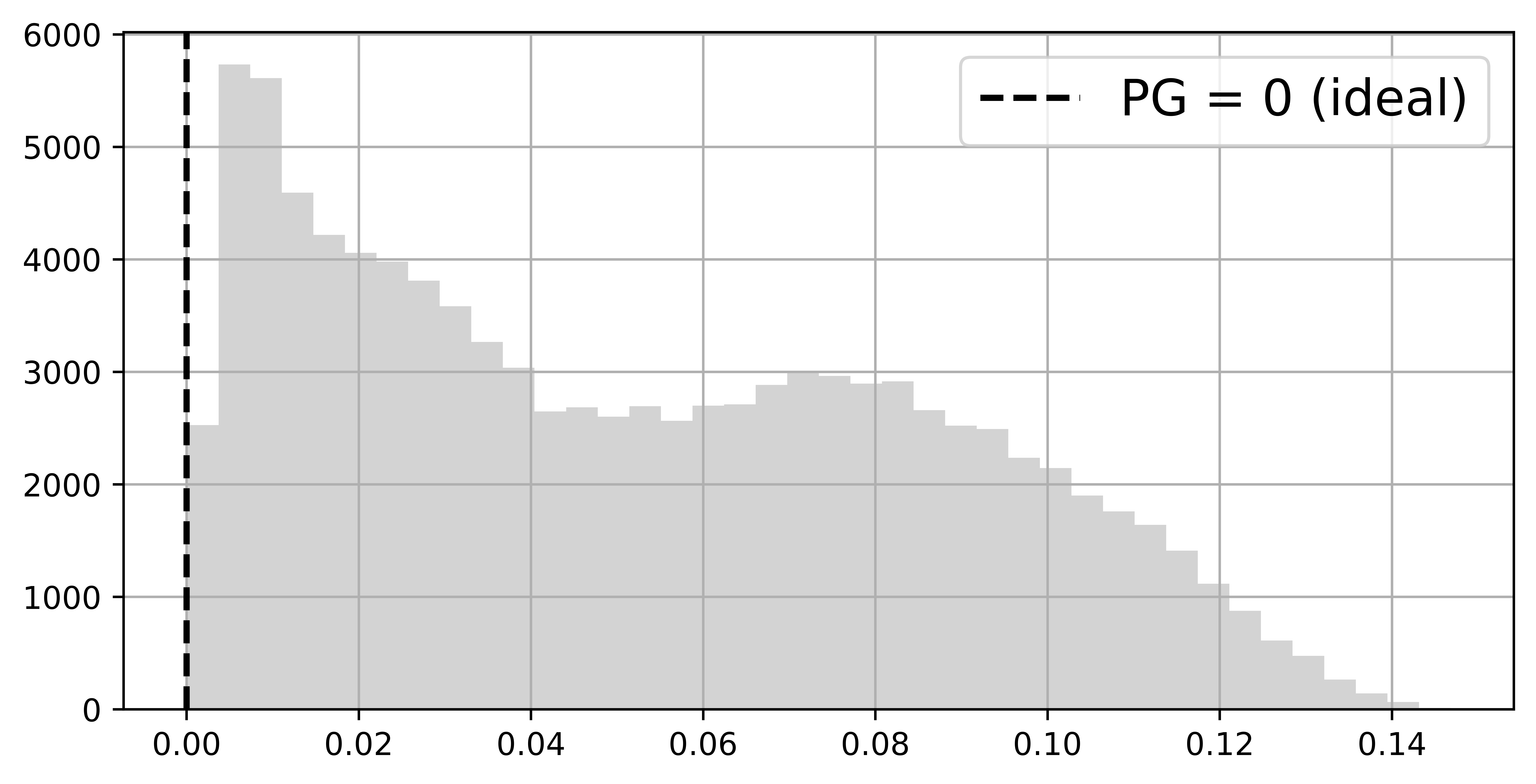}};
\node at (3.5,1.7) {\small $\PG$ score: full ensemble};
\node at (-3.5,-3.3)  {\includegraphics[height=1.15in]{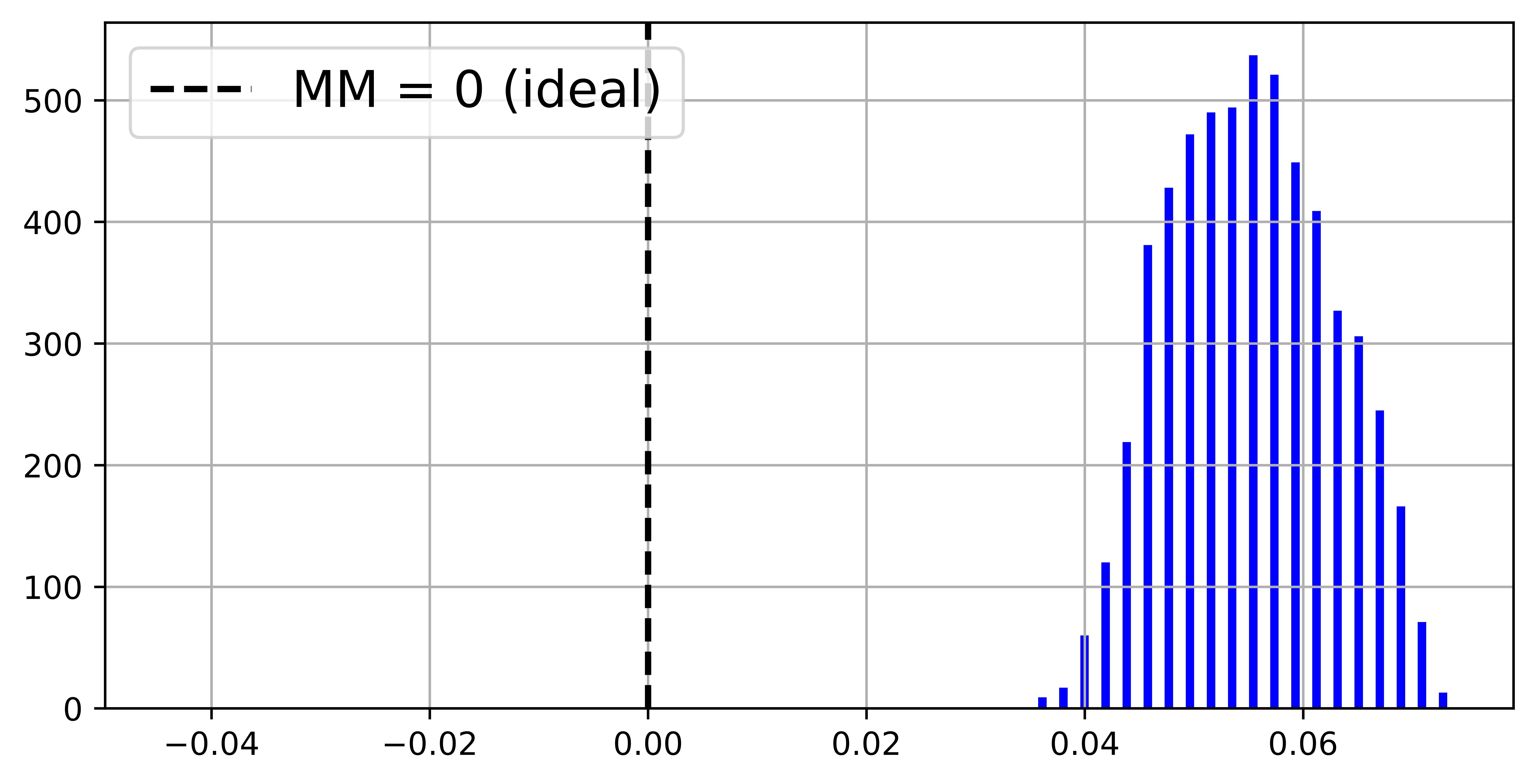}};
\node at (-3.5,-1.7) {\small 5734 most D-favoring};
\node at (3.5,-3.3)  {\includegraphics[height=1.15in]{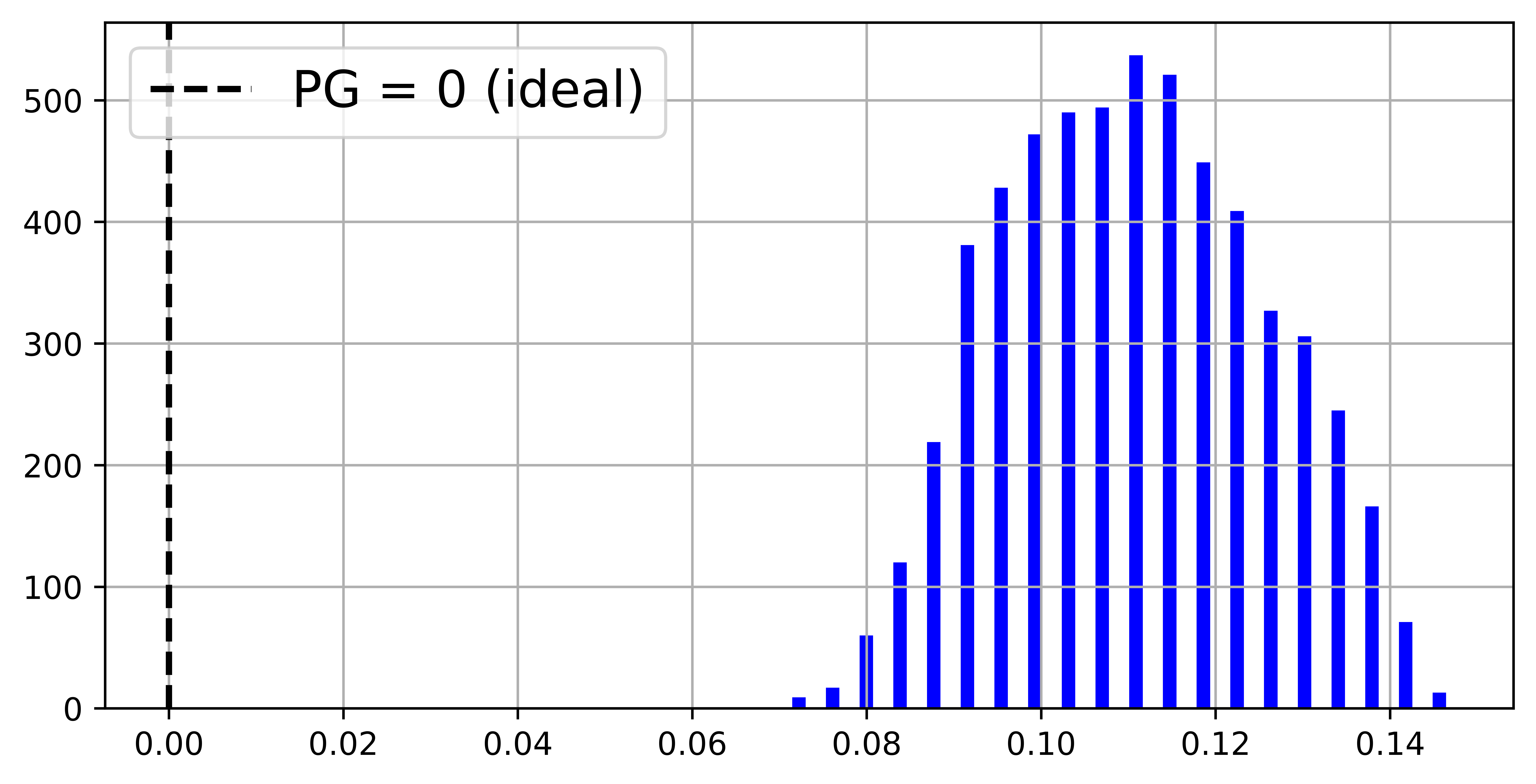}};
\node at (3.5,-1.7) {\small 5734 most D-favoring};
\end{tikzpicture}
\caption{Ensemble outputs for 100,000 Utah Congressional plans with respect to SEN16 votes. Republicans received 71.55\% of the two-way vote in this election, which is marked in the plots to show the corresponding seat share. There are 5734 plans in the ensemble in which Democrats get a seat; these are shown in blue in the top row, but they are absent from the next two rows because a D seat never occurs in plans with good symmetry scores. The last row of the figure shows the $\MM$ and $\PG$ histograms restricted to the plans with a D seat. The empirical data corroborates the prediction that good symmetry scores lock out Democratic representation, and illustrates the "Utah paradox" that a Democratic-won seat always receives the label of a Republican gerrymander.}
  \label{fig:ut_paradox}
\end{figure}
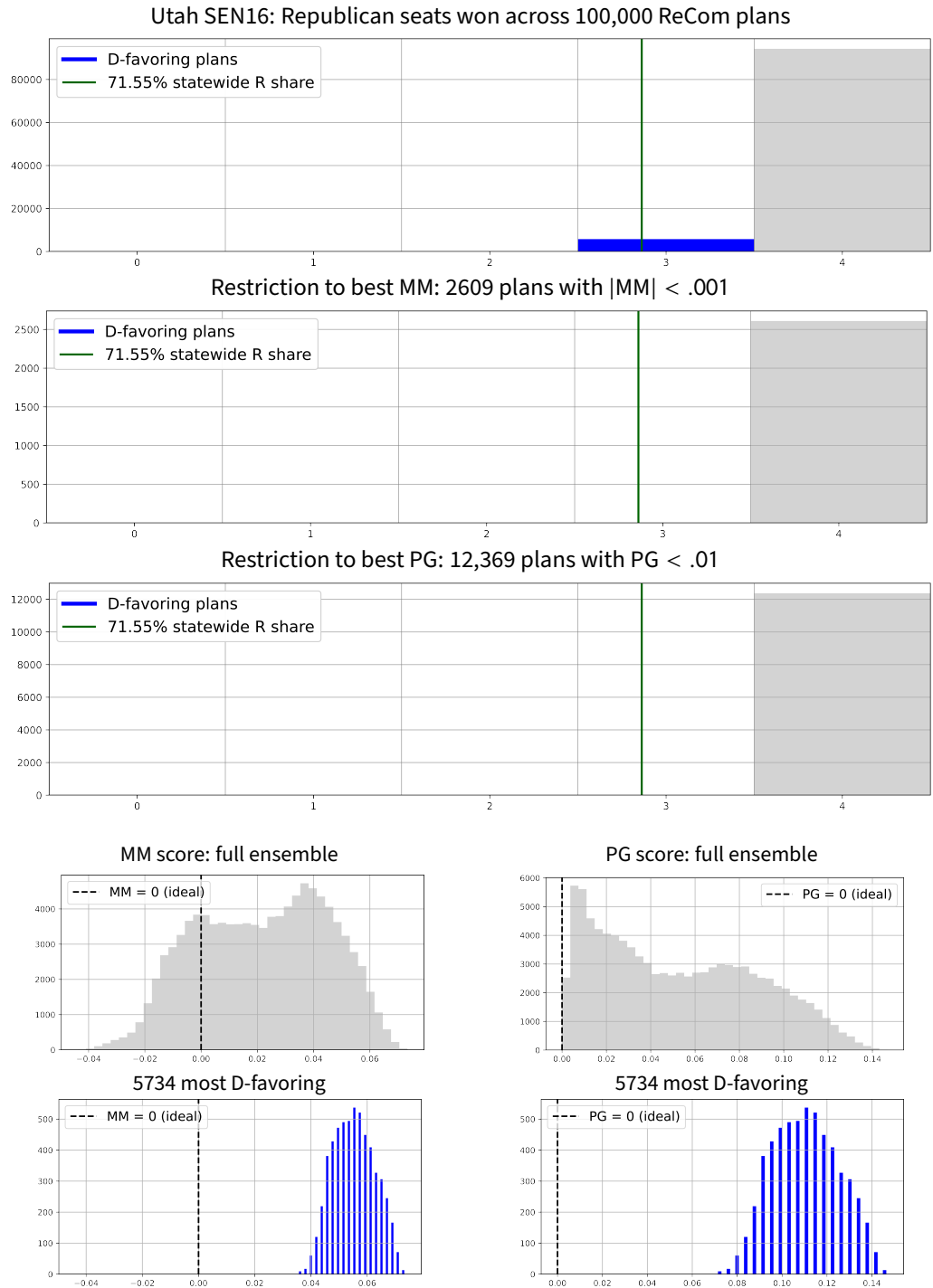

Figure~\ref{fig:ut_paradox} shows outcomes from our
100,000-step ensemble.
The vast majority ($94.266\%$) of Utah plans found in our ensemble have all four R seats, 
with the remaining plans giving 3-1 splits.  
The chain found 5734 plans with  3 Republican and 1 Democratic seats, and we see that all of these have $\PG$ scores  above 0.06. 
Below, we explore and explain these bounds on seats
and scores.

When looking at the full $\PG$ histogram, 
we see a large bulk of plans with nearly-ideal $\PG$ scores, all giving a Republican sweep (four out of four R seats).
This is surprising enough to deserve a name.

\medskip

\centerline{\bf The Utah Paradox}
\begin{itemize}
\item Partisan 
symmetry scores near zero are supposed to indicate
fairness, and signed symmetry scores are supposed
to indicate which party is advantaged.
\item There are many trillions of valid Congressional 
plans in Utah, and under reasonable geographical assumptions, every single one of them
with $\PG$ close to zero is mathematically guaranteed to yield a Republican sweep of the seats.  
In particular,
even constraining symmetry scores to better than the ensemble average (for any reasonably diverse neutral ensemble of alternatives) would 
deterministically impose a partisan outcome:  the 
one in which Democrats are locked out of representation.
\item Furthermore, the signed scores make a sign error:  they report
all plans with Democratic representation to be 
significant pro-Republican gerrymanders.
\end{itemize}

\bigskip

\noindent {\em Geographic assumptions:}
The UT-SEN16 election has a statewide R share of $.7155$, so this is roughly equal to the district mean $\vbar$ (or exactly equal in the equal-turnout case).  If we can show that the possible Republican share of a district is bounded above by any $Q<.931$, then the arguments of the last section show that Democrats can secure at most one seat, and that every plan with Democratic representation has the sign error $\MM>0$.
We consider the assumption that no district can exceed $93\%$ Republican share to be very reasonable.
Indeed, even a greedy assemblage of the 608 precincts with the highest Republican share in that race (which is the number needed to reach the ideal population of a Congressional district) only produces a district with R share $.888$.  And this is even without imposing a requirement that districts be contiguous, which certainly limits the possibilities further and only strengthens the bound.  As a further indication,
our Markov chain run of contiguous plans never encounters a district with Republican share over $.8595$.

\begin{example}[The Utah paradox, empirical]
The UT-SEN16 vote pattern can be divided into 
4R-0D seats or 3R-1D seats.
However, even though $\MM$, $\PB$, and $\PG$ scores
can all get arbitrarily close to zero, there 
are no reasonably symmetric plans that secure a Democratic seat.
In our algorithmic search, every plan with nonzero Democratic representation 
has $\PG>.069$, $\MM>.034$, and $\PB\ge .25$, which is in the worst half of scores observed for each of those scores.  
Thus even a mild constraint on partisan symmetry has imposed an empirical Democratic lockout. As predicted by the analytic paradox, 
all plans with D representation are reported as significant R-favoring  gerrymanders.
\end{example}

The basic idea here is extremely simple and readers can try it for themselves.  Choose any four numbers from 0 to 1 whose mean is $.7155$.  If one of them is at or below $.5$ (a Democratic-won seat), you will find that the median of the four scores is greater than the mean, so the mean-median score is positive.  It is simply false that a median higher than the mean is a flag of advantage for the point-of-view party.

 As described in the introduction,
Utah recently became the first state to encode 
partisan symmetry as a districting criterion in statute.
This makes the Utah Paradox quite a striking example of the worries raised by using partisan symmetry 
scores in practice.

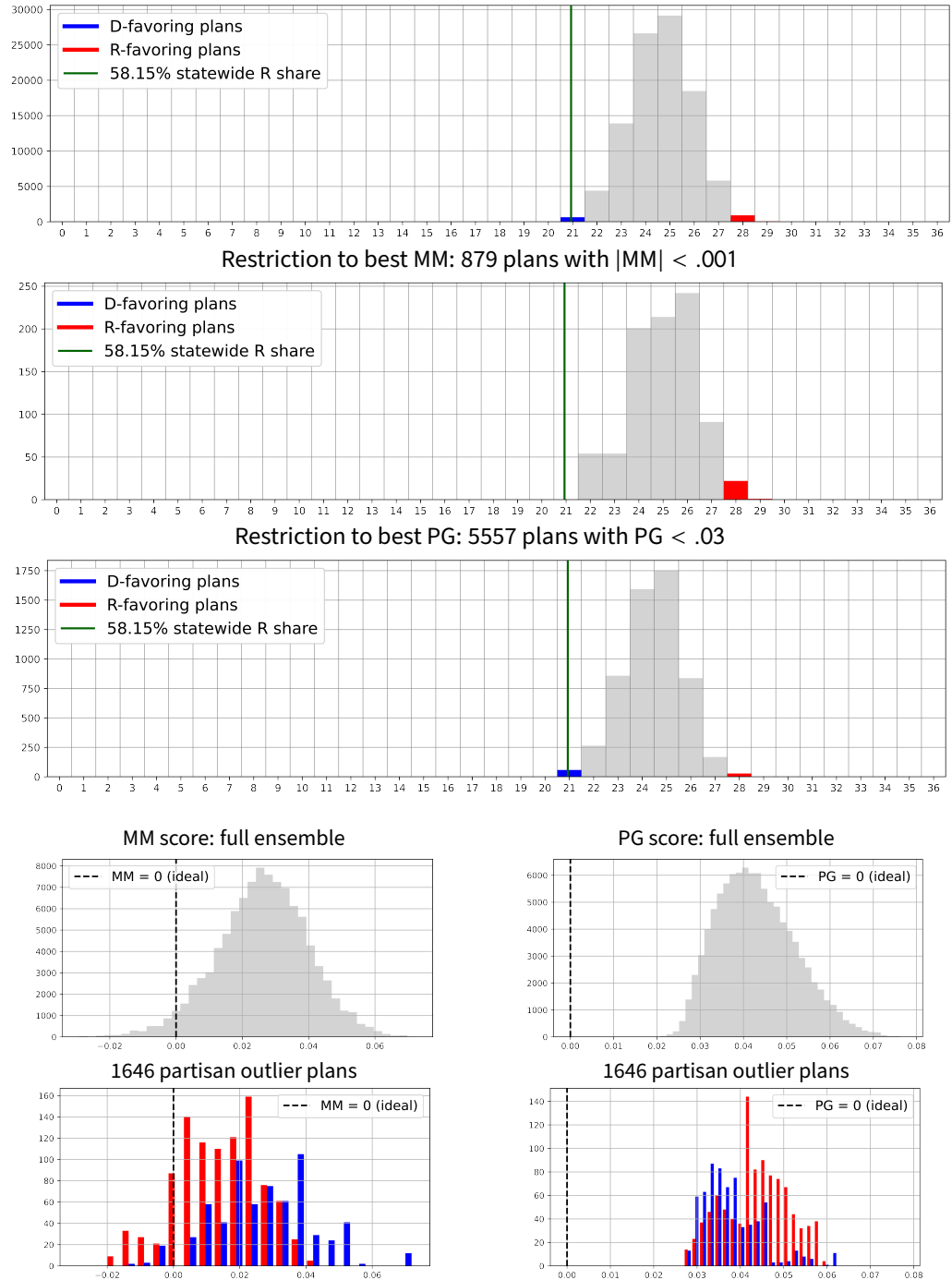
\begin{figure}
\centering
\begin{tikzpicture}
\node at (0,14) {Texas SEN12: Republican seats won across 100,000 {\sf ReCom} plans};
\node at (0,12)   {\includegraphics[height=1.4in]{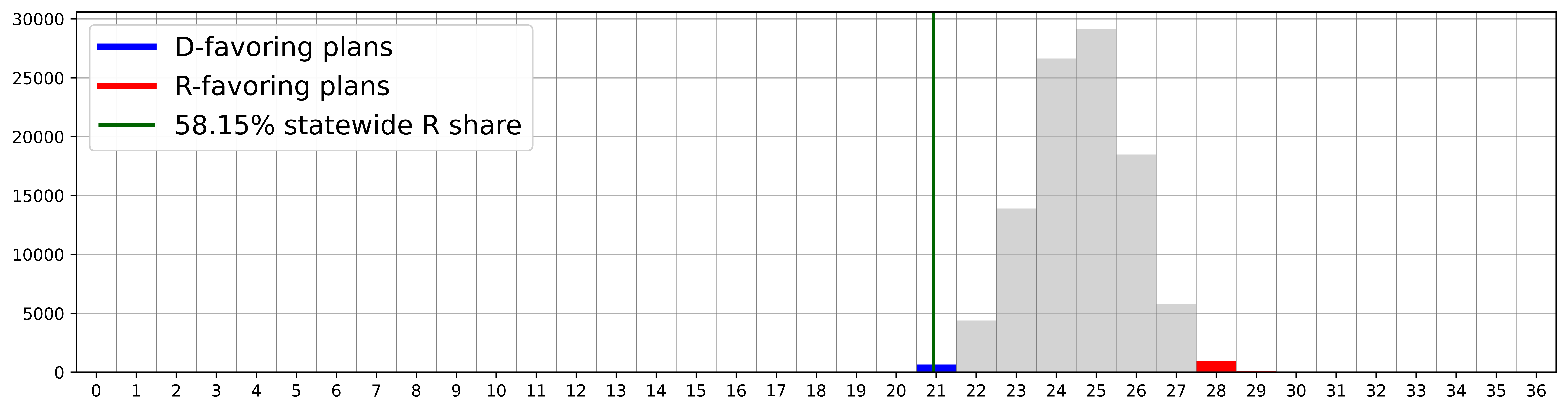}};
\node at (0,10) {Restriction to best $\MM$:
879 plans with $|\MM|<.001$};
\node at (0,8)  {\includegraphics[height=1.4in]{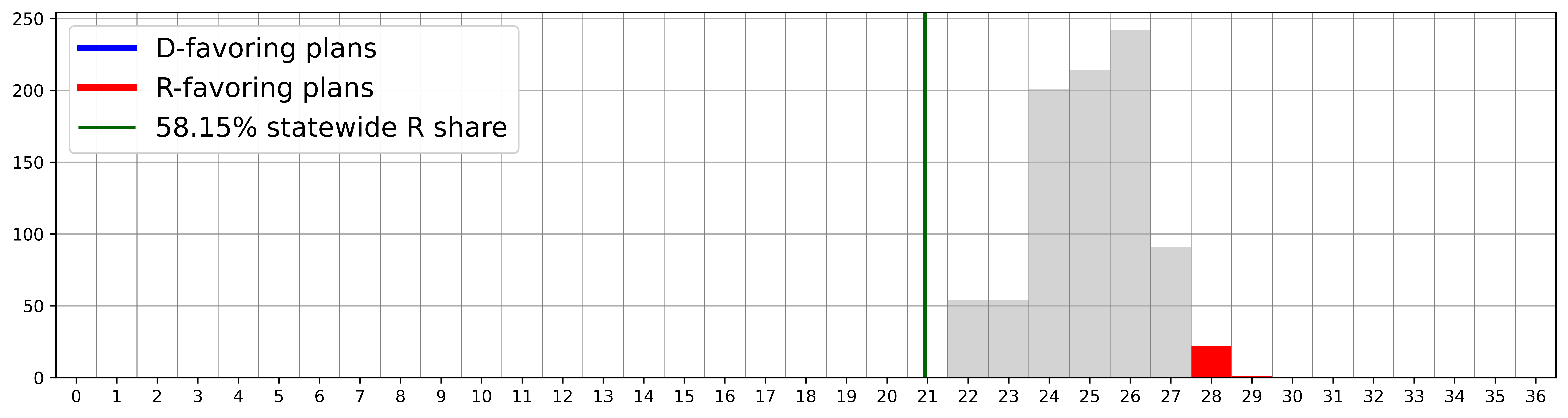}};
\node at (0,6) {Restriction to best $\PG$:
5557 plans with $\PG<.03$};
\node at (0,4)  {\includegraphics[height=1.4in]{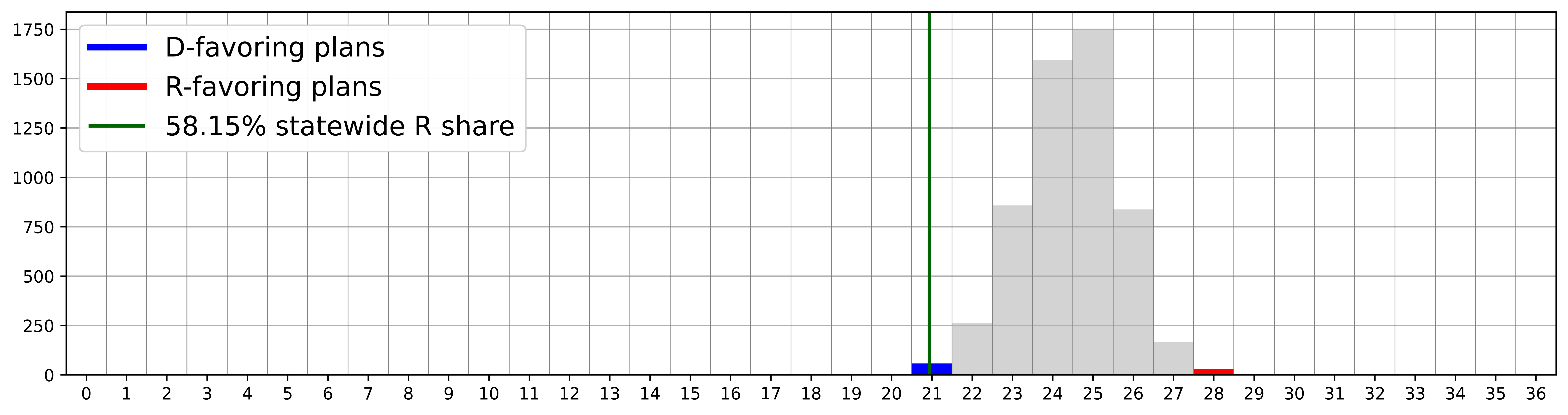}};

\node at (-3.5,0)  {\includegraphics[height=1.15in]{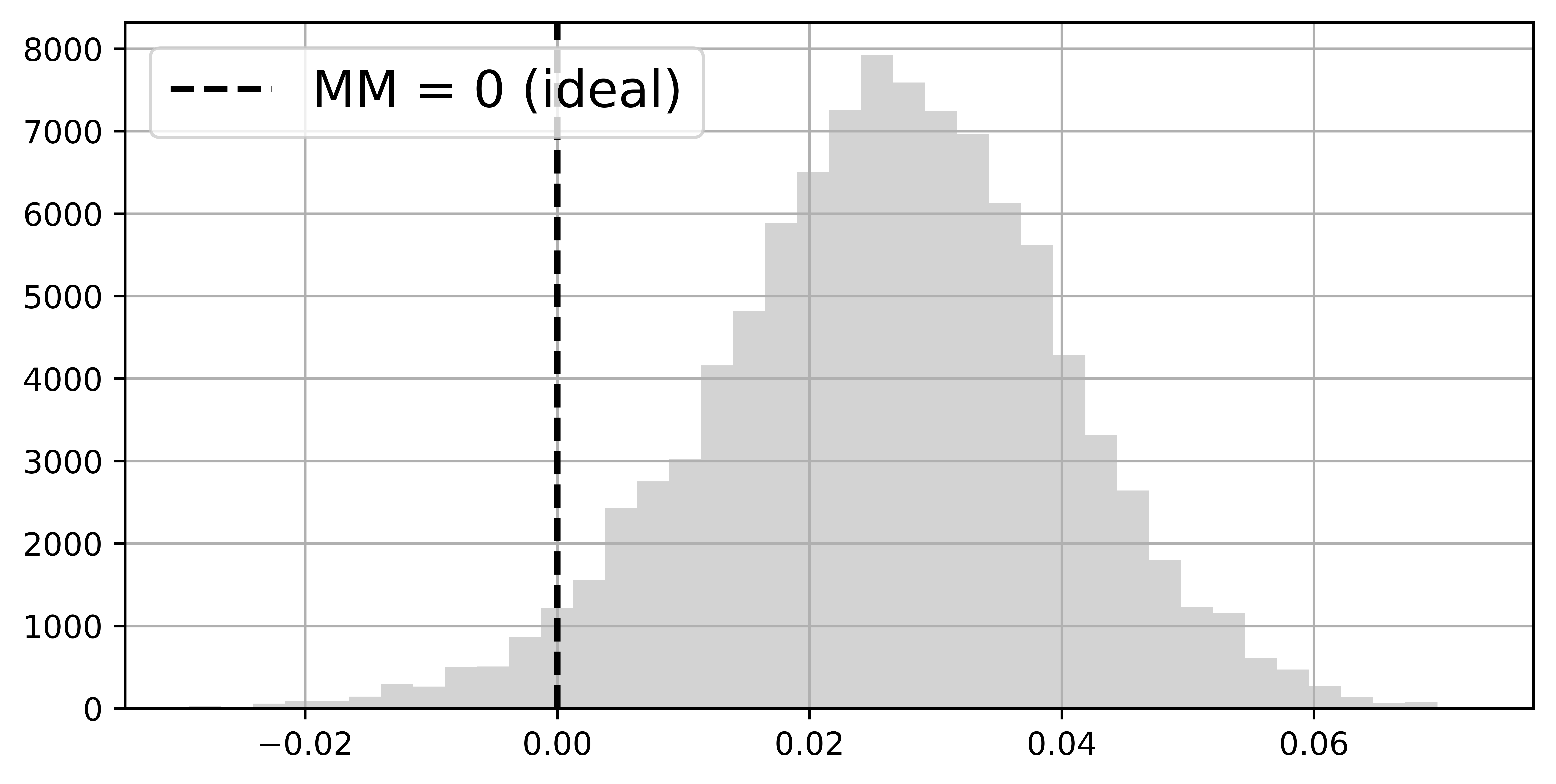}};
\node at (-3.5,1.7) {\small $\MM$ score: full ensemble};
\node at (3.5,0)  {\includegraphics[height=1.15in]{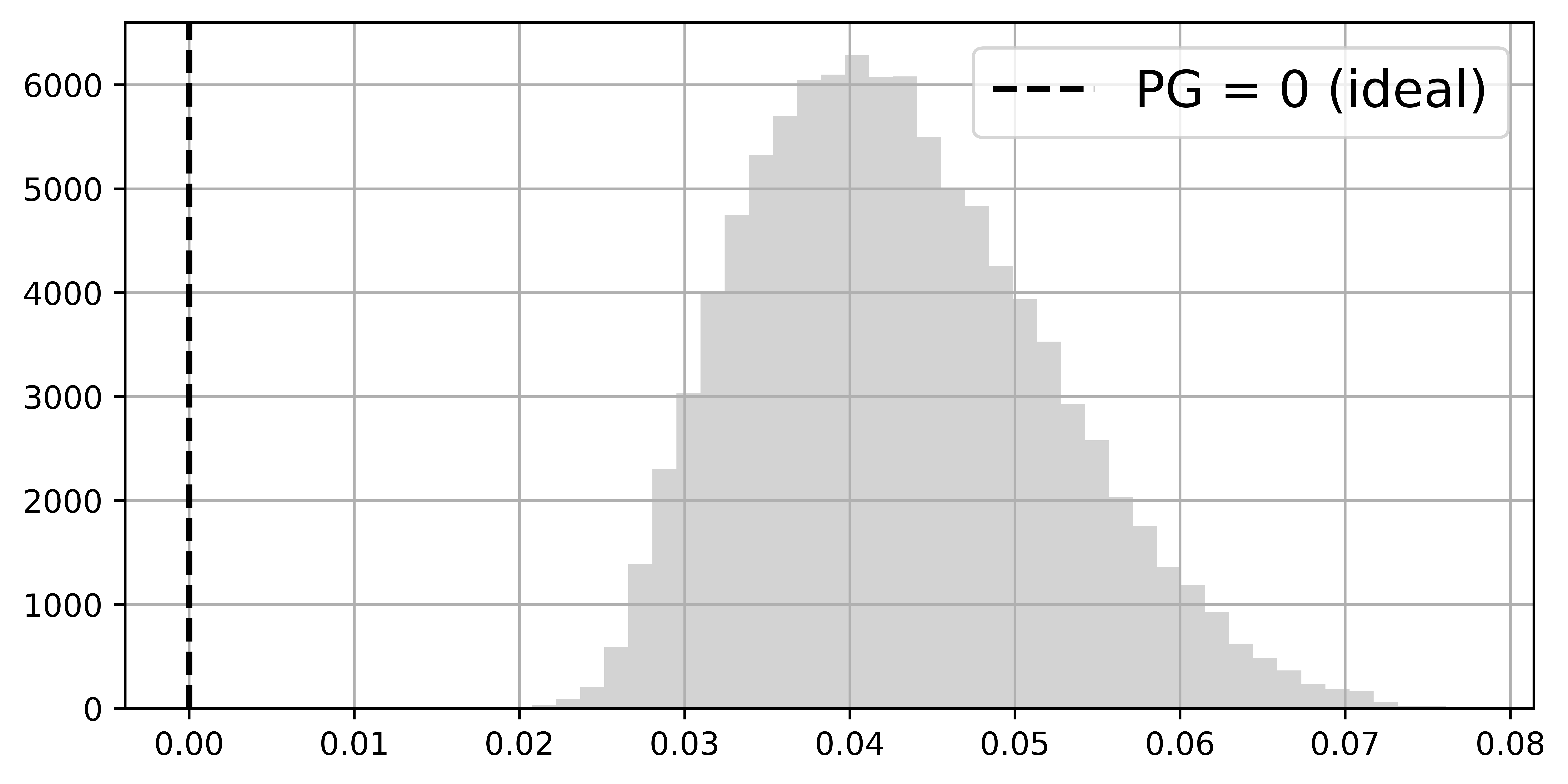}};
\node at (3.5,1.7) {\small $\PG$ score: full ensemble};
\node at (-3.5,-3.3)  {\includegraphics[height=1.15in]{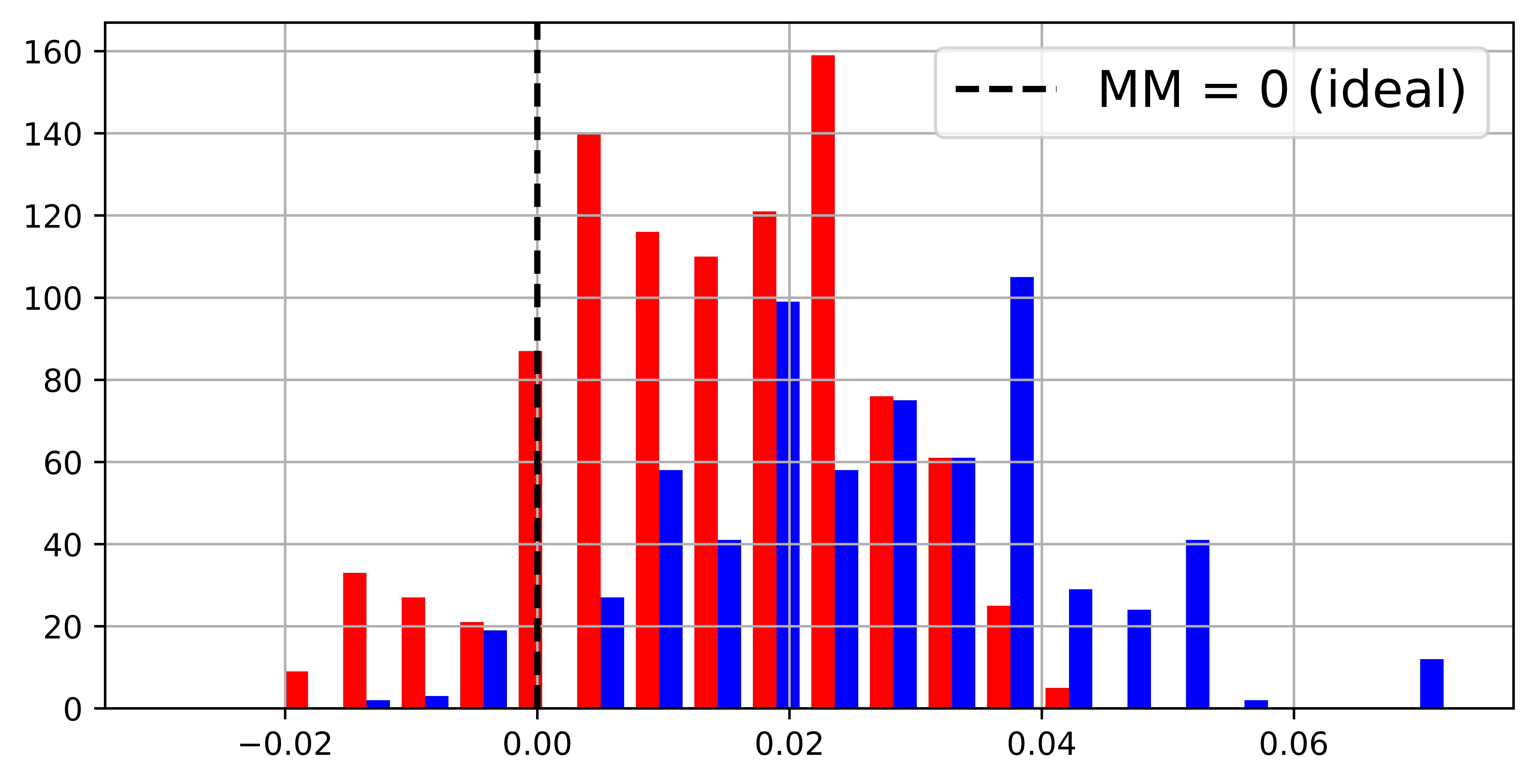}};
\node at (-3.5,-1.7) {\small 1646 partisan outlier plans};
\node at (3.5,-3.3)  {\includegraphics[height=1.15in]{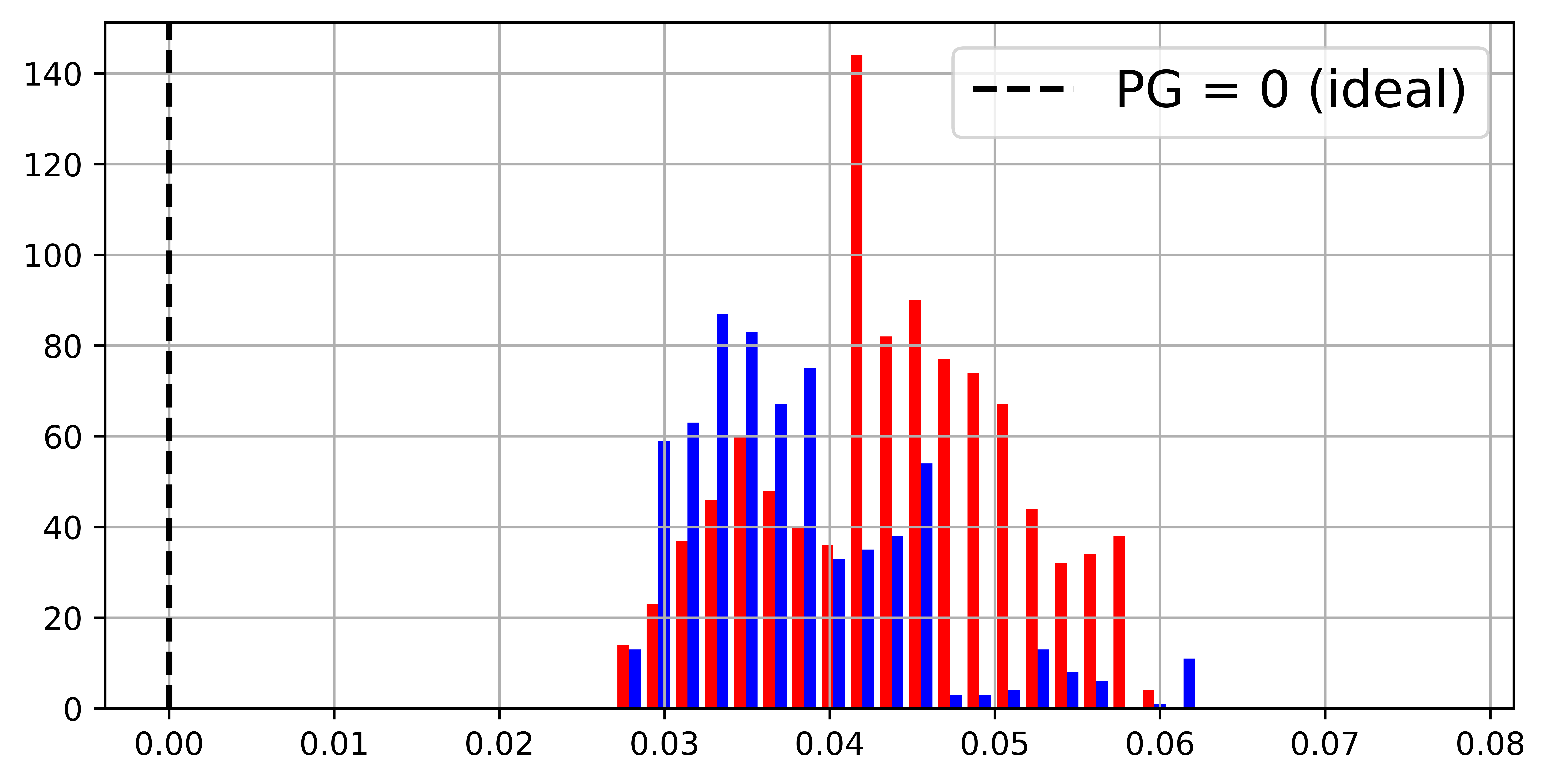}};
\node at (3.5,-1.7) {\small 1646 partisan outlier plans};
\end{tikzpicture}
\caption{Ensemble outputs for Texas Congressional plans with respect to SEN12 votes. Republicans received 58.15\% of the two-way vote in this election, which is marked in the plots to show the corresponding seat share. There are 1646 plans in the ensemble that are seats outliers for one party or the other; these are shown in red and blue in the top row and their relative frequency can be observed in the next two rows, which focus on plans with the best symmetry scores. The last row of the figure shows the $\MM$ and $\PG$ histograms restricted to the 1646 outlier plans flagged above. The scores are shown to be readily gamed: numerous extreme plans are found with near-optimal symmetry scores. In this sample, most extreme Democratic-favoring plans are labeled Republican gerrymanders by the mean-median score, and some extreme Republican-favoring plans are labeled Democratic gerrymanders.}
    \label{fig:texas}
\end{figure}

\begin{figure}
\centering
\begin{tikzpicture}
\node at (0,14) {North Carolina SEN16: Republican seats won across 100,000 {\sf ReCom} plans};
\node at (0,12)   {\includegraphics[height=1.4in]{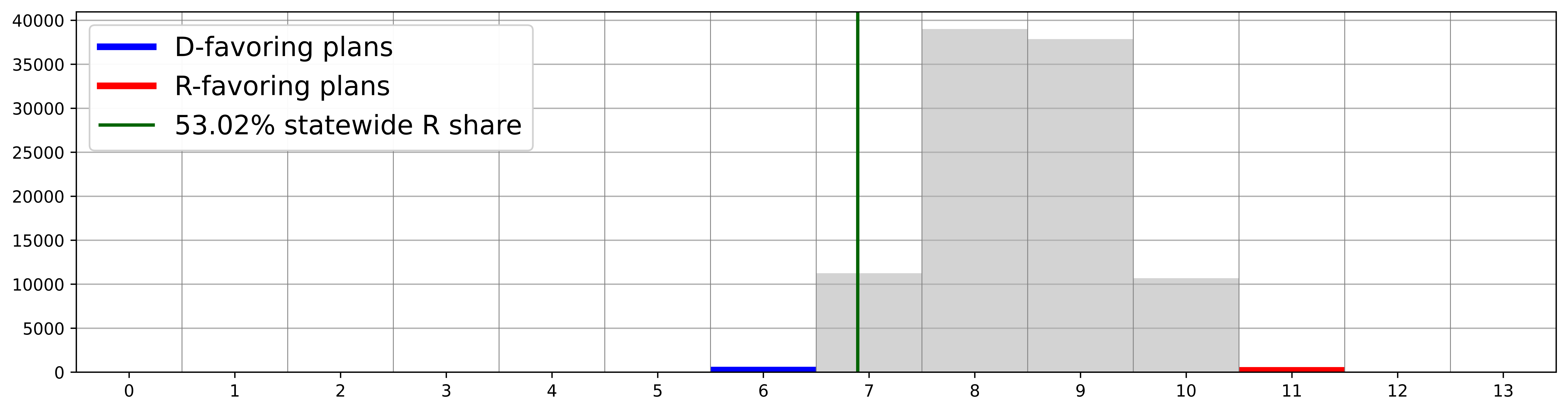}};
\node at (0,10) {Restriction to best $\MM$:
4049 plans with $|\MM|<.001$};
\node at (0,8)  {\includegraphics[height=1.4in]{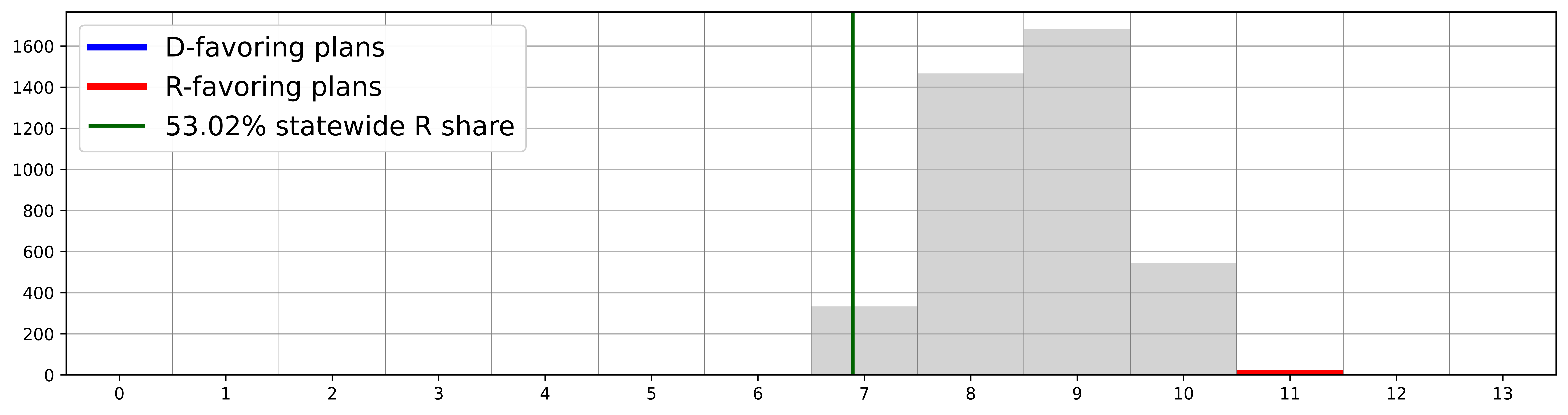}};
\node at (0,6) {Restriction to best $\PG$:
1287 plans with $\PG<.01$};
\node at (0,4)  {\includegraphics[height=1.4in]{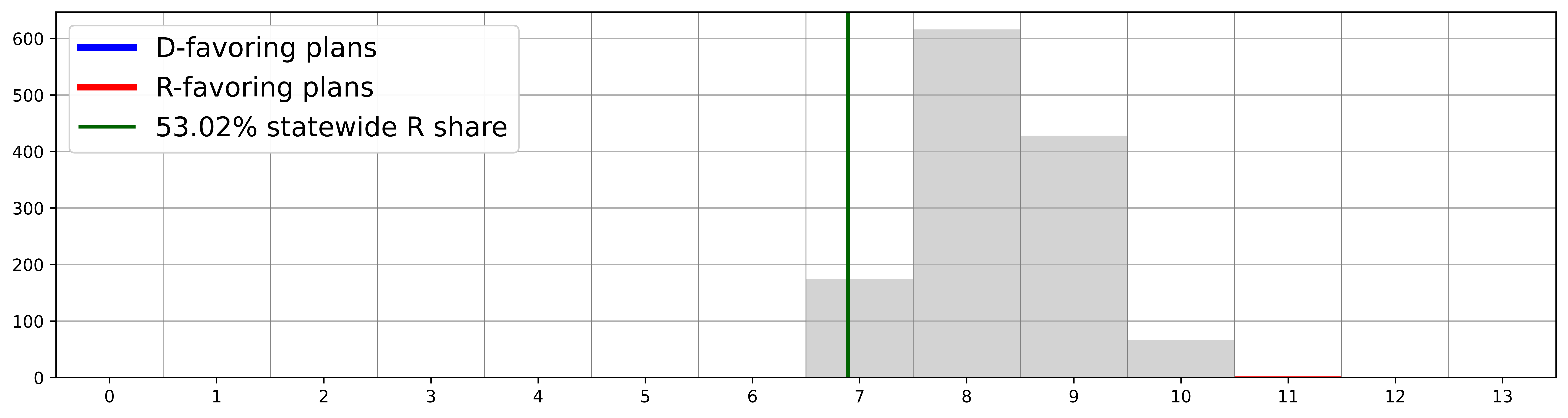}};

\node at (-3.5,0)  {\includegraphics[height=1.15in]{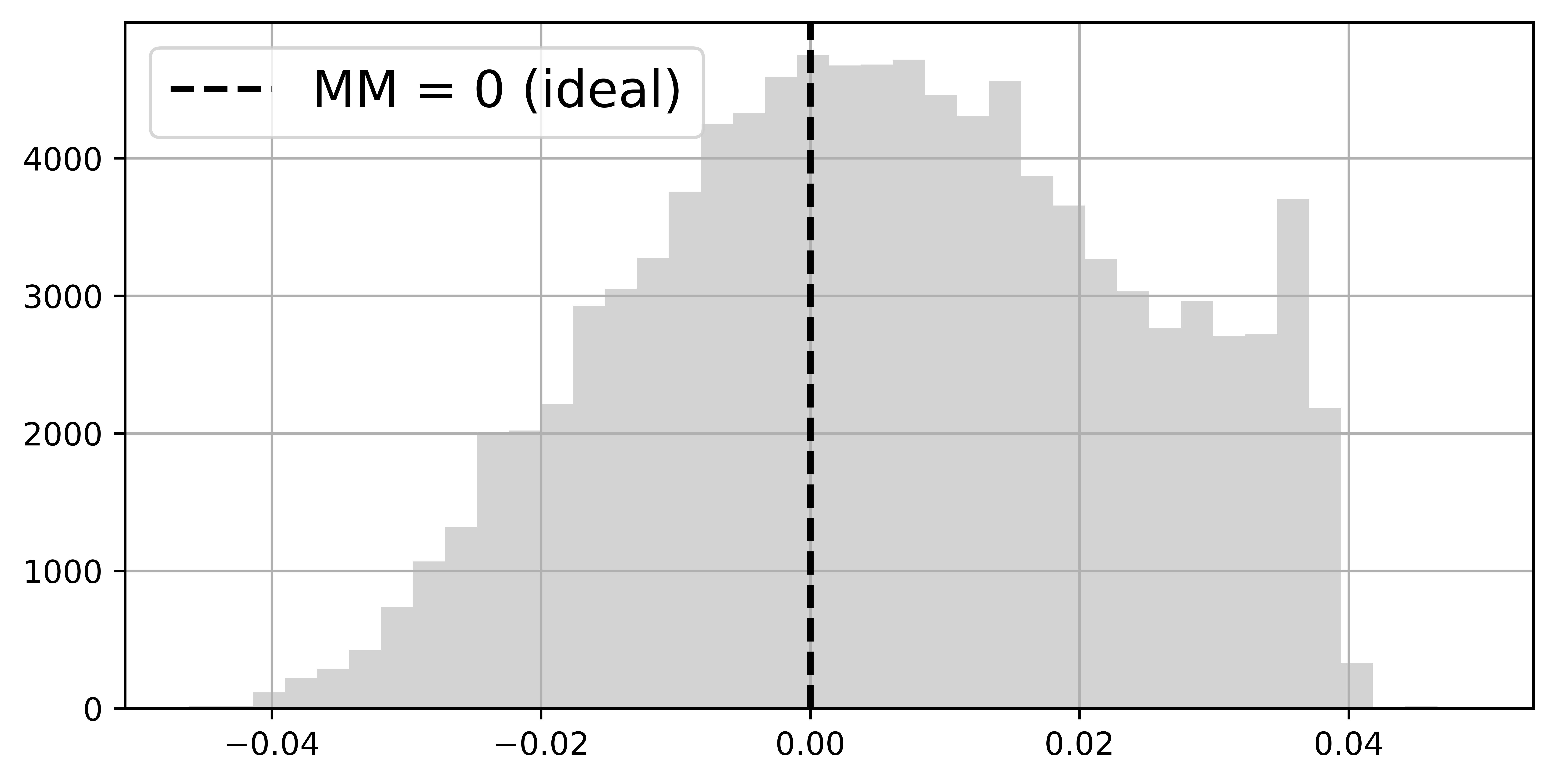}};
\node at (-3.5,1.7) {\small $\MM$ score: full ensemble};
\node at (3.5,0)  {\includegraphics[height=1.15in]{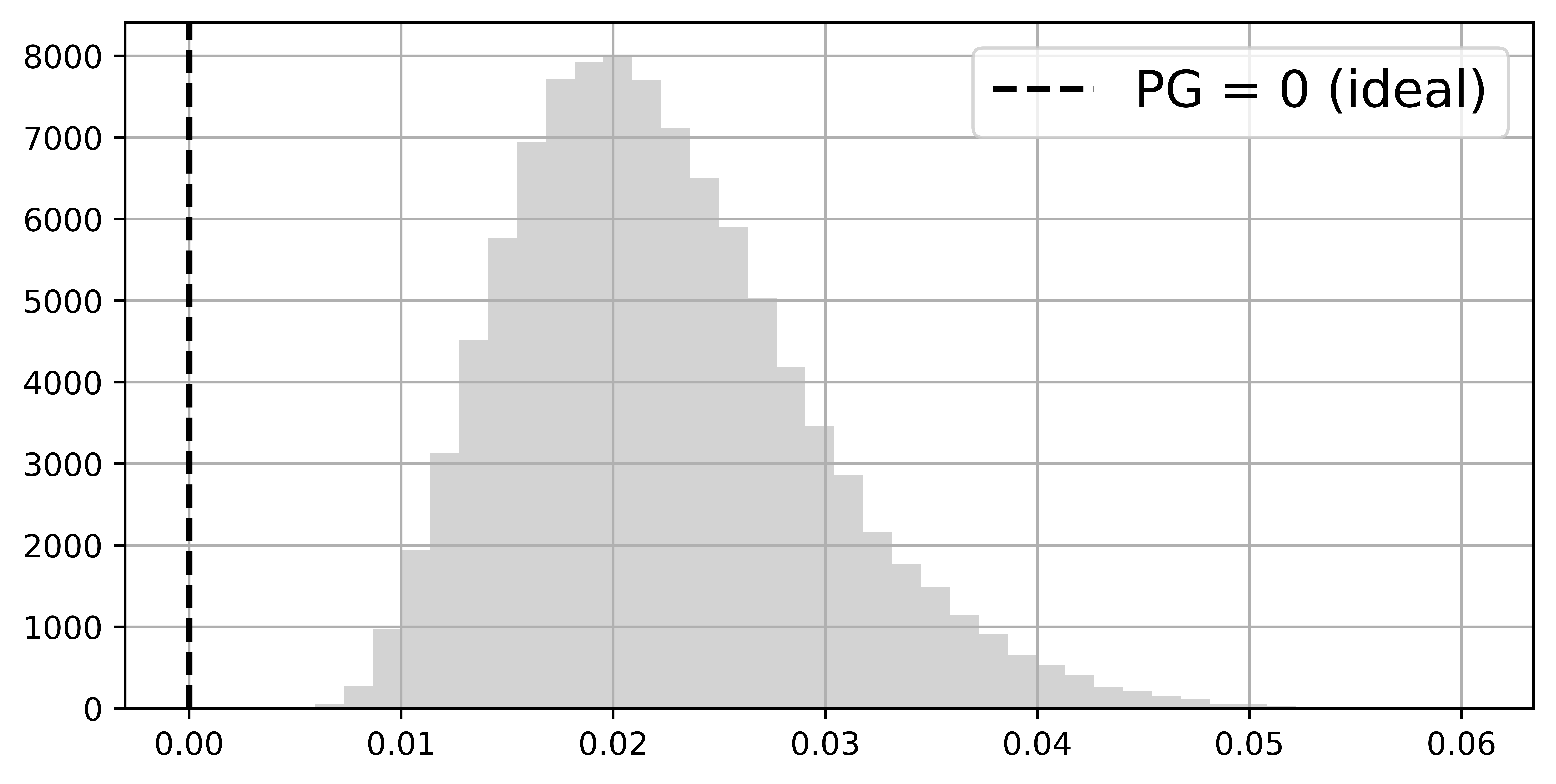}};
\node at (3.5,1.7) {\small $\PG$ score: full ensemble};
\node at (-3.5,-3.3)  {\includegraphics[height=1.15in]{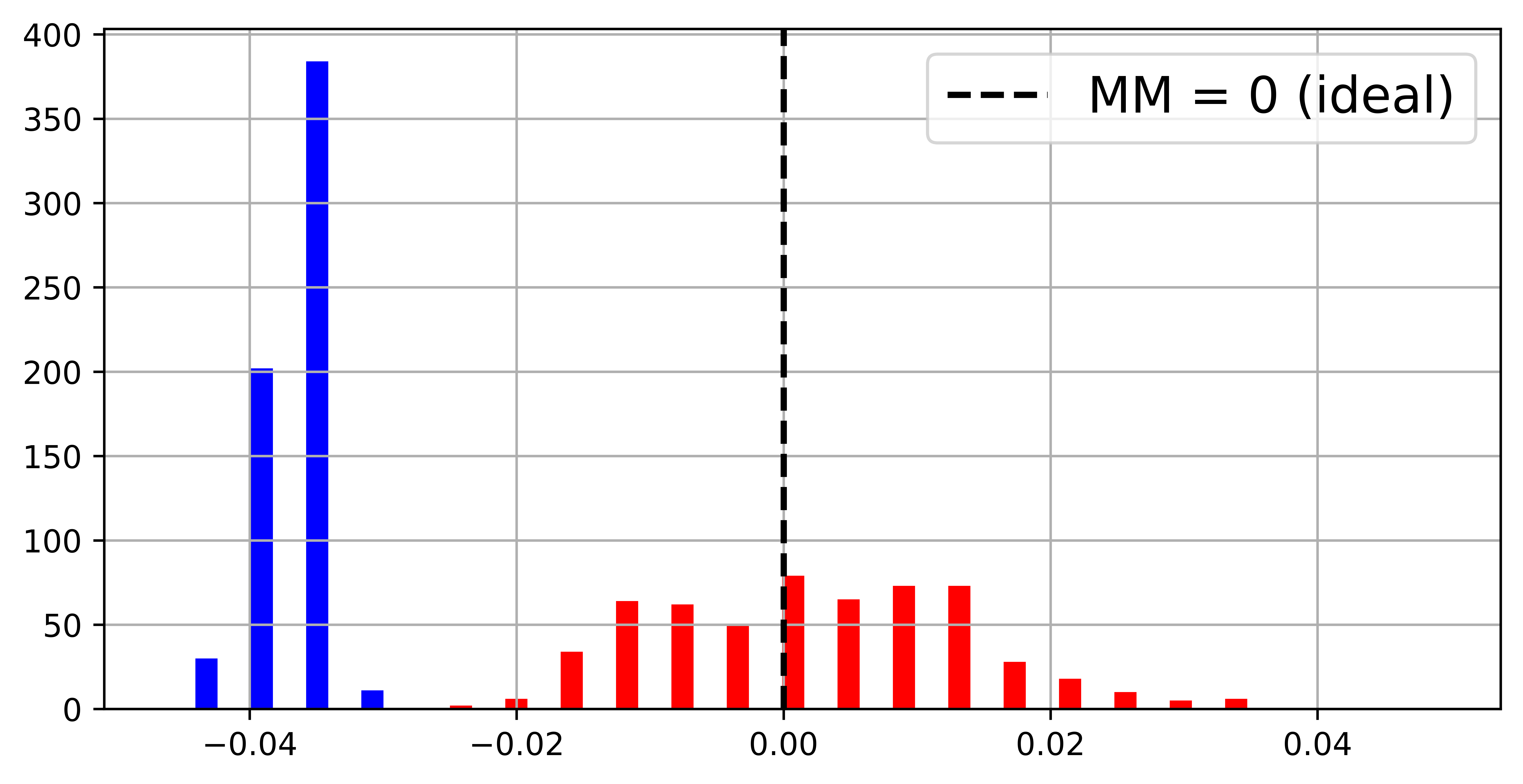}};
\node at (-3.5,-1.7) {\small 1202 partisan outlier plans};
\node at (3.5,-3.3)  {\includegraphics[height=1.15in]{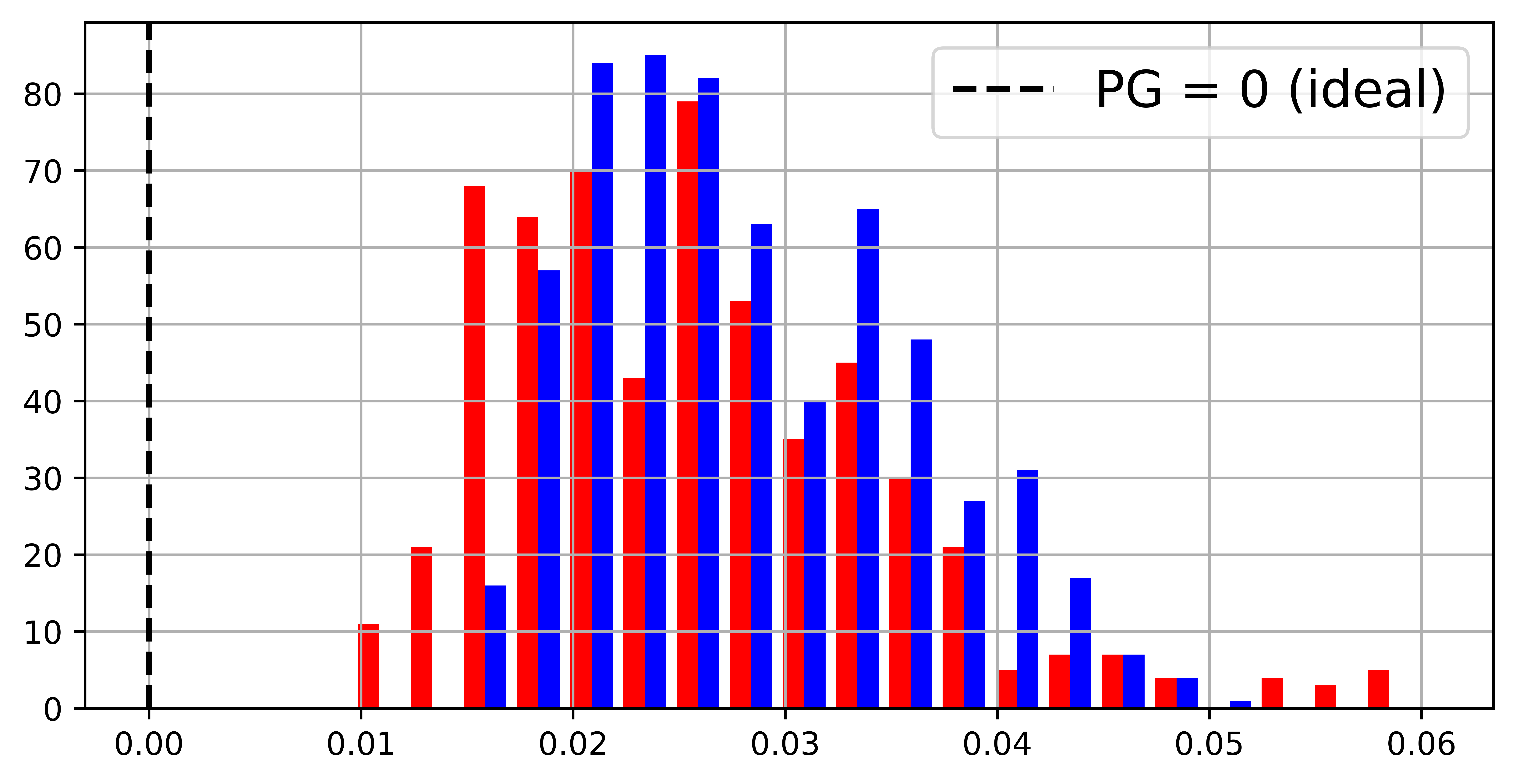}};
\node at (3.5,-1.7) {\small 1202 partisan outlier plans};
\end{tikzpicture}
\caption{Ensemble outputs for North Carolina Congressional plans with respect to SEN16 votes. Republicans received 53.02\% of the two-way vote in this election, which is marked in the plots to show the corresponding seat share. There are 1202 plans in the ensemble that are seats outliers for one party or the other; these are shown in red and blue in the top row and their relative frequency can be observed in the next two rows, which focus on plans with the best symmetry scores. The last row of the figure shows the $\MM$ and $\PG$ histograms restricted to the 1202 outlier plans flagged above. In this setting, symmetry can easily be gamed in favor of Republicans, with thousands of 11-2 plans receiving near-perfect mean-median scores.}\label{fig:NC}
\end{figure}

\subsection{Texas}
Next, we turn to Texas, creating a chain of 100,000 steps to explore the ways to divide up the 2012 Senate vote distribution.  With 36 Congressional districts, Texas has one of the highest $k$ values of any state (only California has more seats). The 2012 Senate race was won by a Republican with $\sim\! 58\%$ of the vote. 
Figure~\ref{fig:texas} shows the partisan properties in the ensemble of plans, allowing us to compare extreme symmetry scores (an ostensible indicator of partisan unfairness) to extreme seat shares (the explicit goal of partisan gerrymandering). We find no evidence of correlation or any kind of correspondence.

Over 98\% of the sampled plans give Republicans  22 to 27 seats out of 36, seen in gray in the histogram. The red bars mark the outlying plans with the most Republican seats (28 or more R seats), while the blue bars mark the  most Democratic plans (21 or fewer R seats). 
We can then study the histograms formed by the winnowed subsets of the ensemble with the best $\PG$ and $\MM$ scores, which in each case fall in the top 6\%. Note that these severely
winnowed subsets not only have a shape similar to the 
full ensemble (indicating a lack of correlation), but that there are still many partisan outlier plans even with strict symmetry standards in place. 
Plans with the extreme outcome of $\ge 28$ R seats actually occur with  {\em higher} frequency among the  $\MM\approx 0$ plans
than in the full sample---more than twice as often, in fact.  This shows rather emphatically that restricting to "good" symmetry scores is no impediment to partisan gerrymandering.

For the reverse perspective, we consider 
how  plans with extreme seat counts score
on symmetry.
The last row in Figure~\ref{fig:texas} shows only the seat outliers: blue for plans with $\le 21$  and 
red for plans with $\ge 28$ Republican seats.  
A significant number of  maximally D-favoring plans (which are also close to vote proportionality) paradoxically register as major Republican gerrymanders under the $\MM$ score, outpacing by a significant margin the most extreme R-favoring plans.
The mean-median score utterly fails at identifying partisan advantage even in an election regarded as "reasonably competitive" by the proponents of partisan symmetry. In Texas, as in Utah, it is simply false that a median higher than the mean is a flag of advantage for the point-of-view party.

In terms of the overall symmetry measured by $\PG$, extreme plans for both parties can be found with 
scores that are as good as nearly anything observed in the ensemble.  
So from this perspective as well,  neither $\MM$ nor $\PG$ signals anything with respect to political outcomes.  
Even if the proponents of symmetry standards never intended to constrain extreme seat imbalances, this runs counter to
the common expectations of anti-gerrymandering reforms in popular discourse, in legal settings, and even in much of the political science literature.

\subsection{North Carolina}
Finally, we move to a state with a much closer to 
even partisan split:  North Carolina ($k=13$ seats),
with respect to the 2016 Senate vote
($\sim \! 53\%$ Republican share).
In this case, mean-median does much better than in Texas
in terms of distinguishing the seat extremes:  Figure~\ref{fig:NC} shows consistently higher
scores for the maps with the most Republican seat
share than  the ones with the most Democratic
outcomes.  However, the extreme Republican maps still straddle the "ideal" score of $\MM=0$, and both sides can still find very extreme plans whose $\PG$ scores report that their symmetry is essentially as good as anything in the ensemble.

Overall it is fair to say that 
partisan symmetry imposes no constraint on partisan gerrymandering in North Carolina, at least for one side:  this method easily produces hundreds of maps with 10-3 outcomes (which was clearly reported in the Rucho case to be the most extreme that the legislature thought was possible) while securing nearly perfect symmetry scores, and the gerrymanderer only needs {\em one}.  
Indeed, the ensemble even finds highly partisan-symmetric maps that return an 11-2 outcome for this particular vote pattern. Four of these are shown in Figure~\ref{fig:11-2}.

\begin{figure}[ht]
    \centering
    
\begin{tikzpicture}
\begin{scope}[yshift=3cm]
\clip (-3.1,-1.3) rectangle (3.5,1.3);
\node at (0,0) {\includegraphics[width=3.5in]{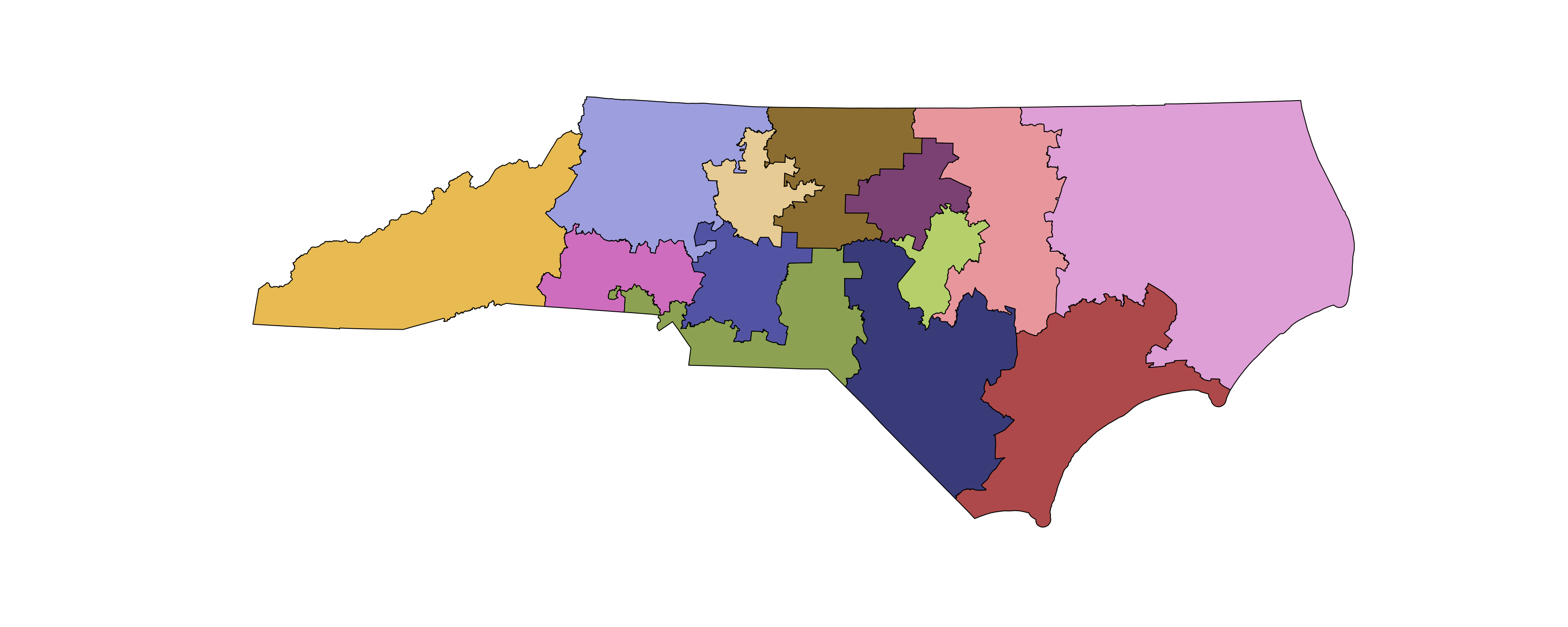}}; 
\end{scope}

\begin{scope}[xshift=8cm,yshift=3cm]
\clip (-3.1,-1.3) rectangle (3.5,1.3);
\node at (0,0) {\includegraphics[width=3.5in]{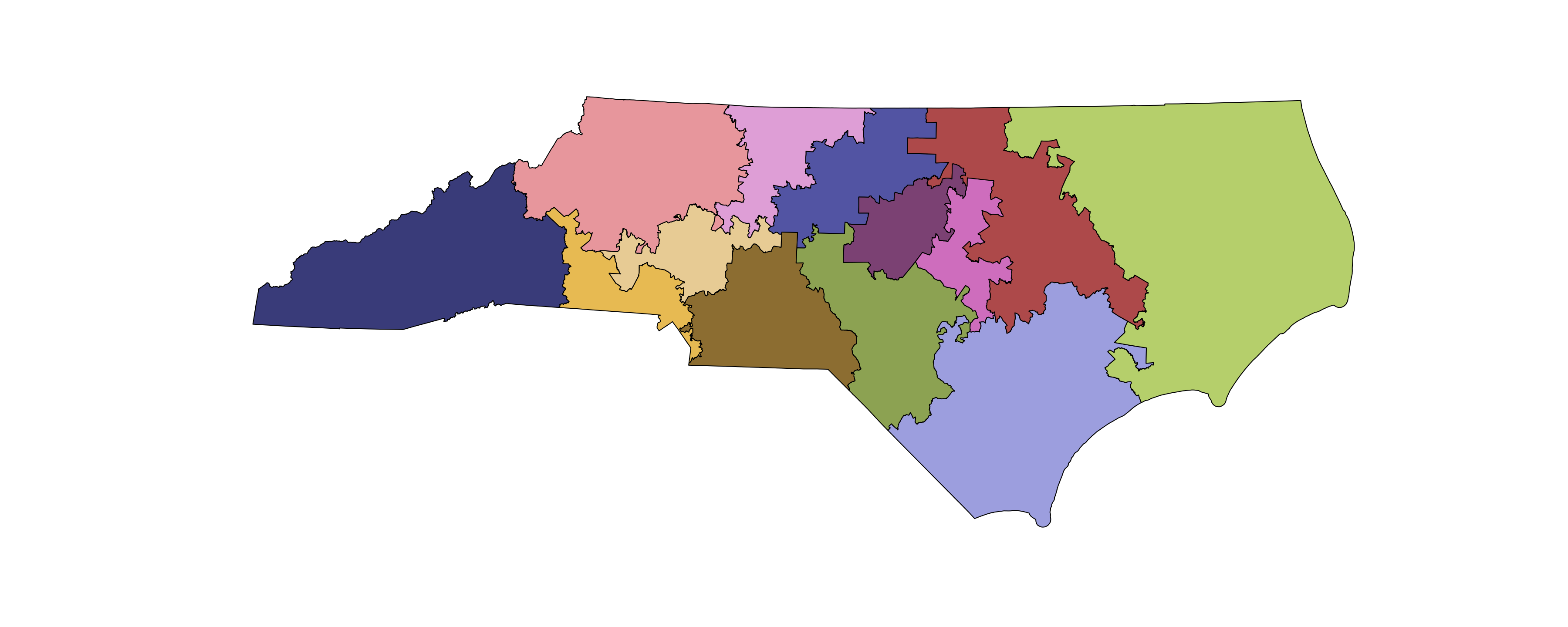}}; 
\end{scope}

\begin{scope}
\clip (-3.1,-1.3) rectangle (3.5,1.3);
\node at (0,0) {\includegraphics[width=3.5in]{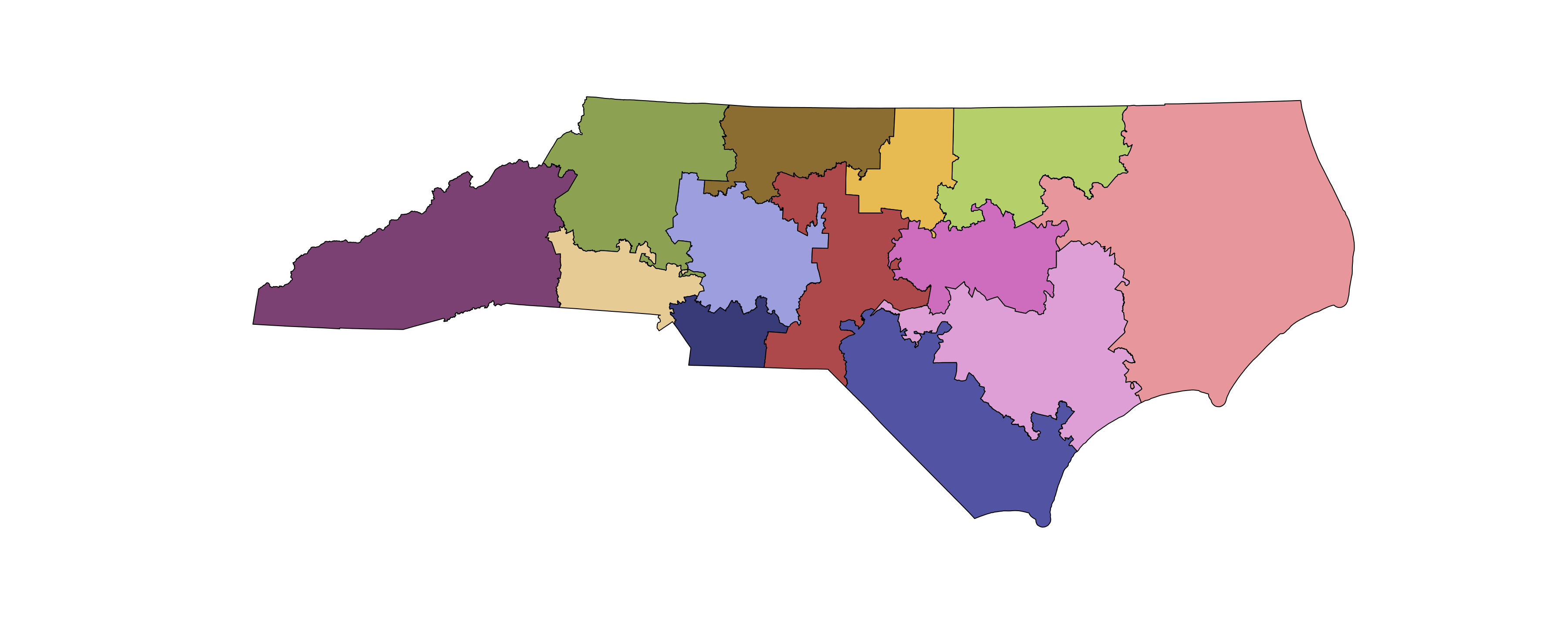}}; 
\end{scope}

\begin{scope}[xshift=8cm]
\clip (-3.1,-1.3) rectangle (3.5,1.3);
\node at (0,0) {\includegraphics[width=3.5in]{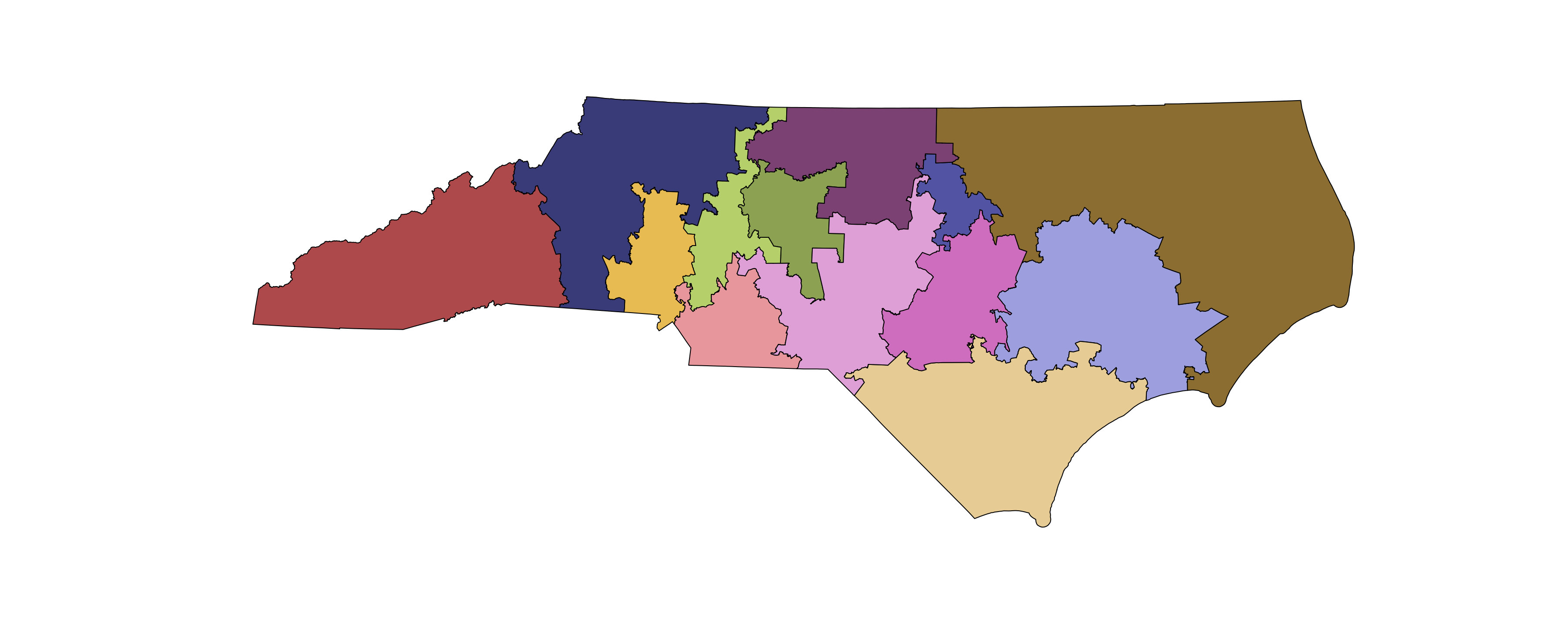}}; 
\end{scope}    
\end{tikzpicture}    
    
\caption{Algorithmic methods are mainly used here for example generation.  Each of these 13-district plans comes out 11R-2D with respect to the SEN16 voting data, while having nearly perfect partisan symmetry.  These maps have $\PG$ scores of 0.0096, 0.0099, 0.0107, and 0.0115, respectively, all in the best 2\% of the ensemble.  This figure also illustrates the diversity of districting plans achieved by this Markov chain method.}
    \label{fig:11-2}
\end{figure}
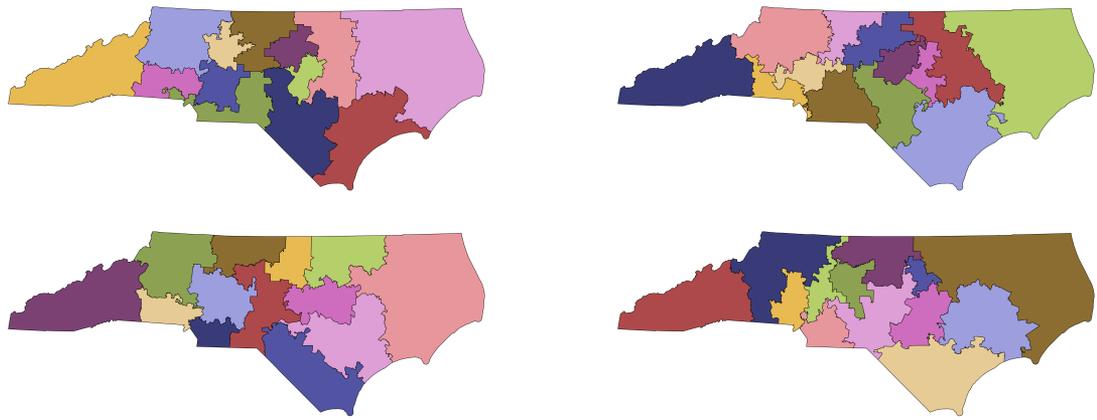

To sum up, we give a recipe for how to hide your partisan gerrymander from detection by symmetry scores, even in a 53-47 state.  To begin, leverage differential turnout.  A recent analysis of 2014-16 Congressional voting showed that several states have voter turnout that is 25 or even 40\% higher in the districts won by one party than the other \citep{Veomett}.
This can easily push the average vote share in the districts two points higher than the statewide share.  Now give your party a 55\% majority in many districts, and let the others be landslides, being sure to arrange the vote shares symmetrically around 55\%.  This secures sterling symmetry scores and a windfall of seats for your side without even risking any close contests.\footnote{Example:  Suppose Party A has 53\% of the statewide vote in a ten-district state, and suppose the average turnout in A-won districts is 80\% that of B-won districts.  Then Party A can get a 7A-3B outcome with all safe seats by dividing up the votes as $\V=(0.26, 0.26, 0.26, 0.55, 0.55, 0.55, 0.55, 0.84, 0.84, 0.84)$, while acing every symmetry test.}

\section{Conclusion}
In this piece, we have characterized the partisan symmetry standard from \cite{KKR} mathematically:  it turns out to amount simply to a prescription for the arrangement of vote totals across districts (Theorem~\ref{thm:equiv}, Partisan Symmetry Characterization).\footnote{One possible response is to try to preserve the partisan symmetry standard but abandon linear uniform partisan swing in favor of a different way of drawing seats-votes curves.  We have discussed alternatives above in several places. We reiterate that noising the seats-votes curve will only change the precision of the characterization. We also note that no change to the swing assumption impacts the findings on the mean-median score, which is defined as the difference between the mean district vote share and the median. }     
We follow this with examples of realistic conditions under which the adoption of strict symmetry standards not only  (a) fails to prevent extreme partisan outcomes but even (b) can lock in  unforeseen consequences on these partisan outcomes.
Finally, again under realistic conditions, signed partisan symmetry metrics  (c) can plainly mis-identify which party is advantaged by a plan.%
\footnote{While this is beyond the scope of the current paper, there is also every reason to believe that partisan symmetry metrics can (d) give answers that depend unpredictably on which vote pattern is used to assess them: endogenous (imputing uncontested races) or exogenous?  Senate race or attorney general? Most single-score indicators have this problem. This may give us a reason to prefer modeling approaches that incorporate multiple races, but those approaches come with a host of modeling decisions that make them impractical for real-world use.}

None of these findings gives a theoretical reason for rejecting partisan symmetry as a definition of fairness.  A believer in symmetry-as-fairness can certainly coherently hold that symmetry standards do not {\em aim} to constrain partisan outcomes, but merely to reinforce the legitimacy of district-based democracy by reassuring the voting public that the tables can yet turn in the future.
This view casts aside, or holds irrelevant, the standard definition of a partisan gerrymander as a plan designed to maximize the seats for a party.
With this reasoning, we should not worry that Democrats in Utah may for now be locked out of Congressional representation {\em by the symmetry standard itself};
this is still fair because Democrats would enjoy a similar advantage of their own if election patterns were to linearly swing by 40 percentage points in their favor.   

For those who do want to constrain the most extreme partisan outcomes that line-drawing can secure, 
these investigations should serve as a strong caution regarding the use of partisan symmetry metrics, whether in the plan adoption stage or in plan evaluation after subsequent elections have been conducted.

\newpage 
If symmetry metrics measured something that was obviously of inherent value in the healthy functioning of representative democracy, then we might reasonably choose to live with the consequences of the definition, no matter the partisan outcomes.
However, the Characterization Theorem shows that a putatively perfect symmetry score is nothing more and nothing less than a requirement that the vote shares $v_i$ in the districts be arranged symmetrically on the number line (see Figure~\ref{fig:eyeball}).  
Someone who wishes to assert that partisan symmetry is {\em really about} some principle---majority rule, responsiveness, equality of opportunity, etc---would have to explain why that principle is captured by the simple arithmetic of vote share spacing.  
With this framing, it is more difficult to argue that symmetry captures any essential ingredient of civic fairness.



\vspace{.6in}

\section*{Funding, Acknowledgments, Data Availability}
The 2019 Voting Rights Data Institute was generously supported by the Prof. Amar J. Bose Research Grant at MIT and the Jonathan M. Tisch College of Civic Life at Tufts.  MD was partially supported by NSF DMS-1255442.

MD is deeply grateful to Gary King for illuminating conversations about the partisan symmetry standard.  All authors thank the other participants of the  Voting Rights Data Institute, especially Brian Morris, Michelle Jones, and Cleveland Waddell, for stimulating discussions.

The replication materials for this paper can be found at 
\citep{PSymm}. We also make use of the code and data in \citep{gerrychain,MGGG-states}.

\printbibliography{}

\newpage
\appendix

\section{Supplement: Proof of characterization theorem}\label{sec:equiv-proof}

We briefly recall the needed notation from above: vote share vector $\V$ with $i$th coordinate 
$v_i$; gap vector $\D$ with $\delta_i=v_{i+1}-v_i$;
and jump vector $\J$ with $j_i=\frac 12 + \vbar - v_{k+1-i}$, where $\vbar$ is the mean of the $v_i$.
These expressions define $\J,\D$ in terms of $\V$;
neither $\J$ nor $\D$ completely determines $\V$ because they are invariant under translation of the entries of $\V$, but one additional 
datum (such as $v_1$ or $\vbar$)
suffices, with $\J$ or $\D$, 
to fix the associated $\V$.
In this section, we begin by 
expressing $\PG$ in terms of the 
jumps $\J$, then giving equivalent conditions for 
$\PG=0$ in terms of $\J$, $\D$, or $\V$.

As outlined above, \PG measures the area between the seats-votes curve 
$\gamma$ and its  reflection. 
The shape of the region between those
curves depends directly on the points $\J= (j_1, j_2, \dots, j_k)$, since each $j_i$ is the $x$ value of a vertical jump in the curve and the $1-j_i$ are the values
of the jumps in the reflection.  
But looking at Figure~\ref{fig:sv-curve}
makes it clear that it is complicated to 
decompose the integral into vertical
rectangles in the style of a Riemann sum,
because the $\{j_i\}$ and the $\{1-j_{k-i}\}$
do not always alternate.  
Fortunately, it is always easy to 
decompose the picture into horizontal
rectangles (analogous to a Lebesgue
integral), where it is now clear which 
red and blue corners to pair as the 
seat share changes from $i/k$ to $(i+1)/k$.
The curve contains the points 
$(j_i, \frac{i-1}k)$, $(j_i,\frac{i}k)$
as well as
$(j_{k+1-i}, \frac{k-i}k)$, 
$(j_{k+1-i}, \frac{k-i+1}k)$.
The rotated curve therefore contains
the points 
$(1-j_{k+1-i}, \frac {i-1}k)$ and 
$(1-j_{k+1-i}, \frac ik)$, which 
means that the $i$th rectangle has 
height $1/k$ and width 
$\bigl| j_i+j_{k+1-i}-1 \bigr|$.
Summing over these rectangles gives us the expression 
$$\PG = \frac 1k \sum_{i=1}^k \bigl|j_i + j_{k+1-i} - 1\bigr|.$$


Recall that the set of vote share vectors
$\mathcal V$ is the cone in the 
vector space $\R^k$ given by 
the condition that the $v_i$
are in non-decreasing order in 
$[0,1]$.
The $\J$ vector is simply 
the $\V$ vector reversed
and re-centered at $1/2$ rather
than $\vbar$.
The only condition on the gap
vector $\D$ is that its entries are non-negative and
sum to at most one.
Putting these observations together
we may define the set of 
achievable $\V,\D,\J$ respectively
as
$$\mathcal V=\left\{
(v_1,\dots,v_k) :
0\le v_1\le\dots\le v_k\le 1
\right\},$$
$$ \mathcal D
=\left\{(\delta_1,\dots,
\delta_{k-1}): \delta_i\ge 0 \ \forall i, \quad 
\sum \delta_i\le 1
\right\},$$
$$\mathcal J=\left\{(j_1,\dots,j_k):
0\le j_1 \le \dots \le j_k \le 1,
\quad 
\sum j_i=\frac k2\right\}.$$
The condition on $\J$ is of  interest because it exactly identifies the possible seats--votes curves 
$\Gamma = \{\gamma_\V:\V\in\mathcal V \}$.
(That is, not just any step function is realizable as a valid  $\gamma$.)

Now we can prove the Characterization theorem.
\setcounter{theorem}{2}
\begin{theorem}
    Given $k$ districts with vote shares $\V$, jump vector $\J$, and gap vector $\D$, the following are equivalent:
    \begin{align}
    \setcounter{equation}{0}
        &\PG(\V) = 0 \label{1}\\
        &j_i + j_{k+1-i} - 1 = 0 \qquad  \forall  i \label{2} \\
        &\frac 12 \left(v_i + v_{k+1-i}\right) = \vbar \qquad \forall  i  \label{3}\\
        &\frac 12 \left(v_i + v_{k+1-i}\right) = \vmed \qquad \forall  i  \label{4}\\
        &\delta_i = \delta_{k-i} \qquad \forall  i \label{5}
    \end{align}
\end{theorem}

\begin{proof}
The condition that $\PG(\V) = 0$ has been
rewritten in terms of $\J$ above, and  converting back 
to the $v_i$ we get
$$\frac{1}{k}\sum_{i=1}^k |j_i + j_{k+1-i} - 1| = \frac 2k \sum_{i=1}^k \left| \frac{v_i+v_{k+1-i}}2 - \vbar\right|=0,$$
which immediately  gives  \eqref{1} $\iff$ \eqref{2} $\iff$ \eqref{3} since a sum of nonnegative terms is 
zero if and only if each term is zero. 
To see \eqref{3} $\iff$ \eqref{4}, just consider
$i=\ceil{\frac k2}$ in \eqref{3} to obtain $\vmed=\vbar$; 
in the other direction, average both sides over $i$ in 
\eqref{4}
to obtain $\vbar=\vmed$.  
Finally, the symmetric gaps condition \eqref{5} is clearly equivalent to the symmetry of the values
of $\V$ about the center $\vmed$,
which is \eqref{4}.  
\end{proof}

\section{Supplement: Bounding partisan Gini in terms of mean-median}\label{sec:k34} \label{sec:PGMM}

Recall that the mean-median score
$\MM$ is a signed
score that is supposed to 
identify which party has a structural advantage, and by what
amount.  On the other hand, the 
partisan Gini 
$\PG$ is a non-negative score
that simply quantifies the
failure of symmetry, interpreted
as a magnitude of unfairness.  
We easily see that  $\PG=0\implies \MM=0$
by comparing 
\eqref{3} and \eqref{4} in Theorem~\ref{thm:equiv}.
In this section we strengthen that to a bound from below
that is sharp in low dimension.

Let us define 
$\discrep(i)=\frac{v_i+v_{k+1-i}}2-
\vbar$, measuring the difference
between the average of a pair 
of vote shares from the average
of all the vote shares.
This gives
$$\PG=\frac 1k \sum_{i=1}^k |2\vbar
-v_i-v_{k+1-i}|=\frac 2k 
\sum_{i=1}^k |\discrep(i)|.$$

Note that $\discrep(\ceil{\frac k2})=\MM$, as
observed above, 
and that $\sum_{i=1}^k \discrep(i)=0$ by definition
of $\vbar$.

\setcounter{theorem}{6}
\begin{theorem}The partisan Gini score satisfies
$\begin{cases}
\PG\ge \frac 4k |\MM|,& k ~{\rm odd}\\
\PG\ge \frac 8k |\MM|,& k ~{\rm even},
\end{cases}$
with equality when $k=3,4$.
\end{theorem}
\begin{proof}
First suppose $k$ is odd, say $k=2m+1$.
Then $\discrep(m)=v_m-\vbar=\MM$, so 
$\sum_{i\neq m}\discrep(i)=-\MM$.
We have 
\begin{align*}\PG=\frac 2k\sum_{i=1}^k |\discrep(i)|=
\frac 2k \left(|\discrep(m)|+\sum_{i\neq m} |\discrep(i)|\right)&\ge
\frac 2k \left(|\discrep(m)|+\bigl|\sum_{i\neq m} \discrep(i)\bigr|\right)\\
&=\frac 2k\left(|\MM|+|-\MM|\right)=\frac 4k |\MM|.
\end{align*}

The argument for even $k=2m$ is very 
similar, except that 
$\discrep(m)=\discrep(m+1)=
\frac{v_m+v_{m+1}}2-\vbar=\MM$.
So now those terms contribute $2|\MM|$ to 
the sum and the remaining terms contribute 
at least $2|-\MM|$, for a bound of 
$\PG\ge \frac 8k|\MM|$.
That completes the proof of the inequalities.

For $k=3$ or $k=4$, the term $\sum_{i\neq m} |\discrep(i)|$
is just $2|\discrep(1)|$, making the inequality into an equality. 
\end{proof}

By a dimension 
count, it is easy to 
see that $\PG$ is not simply
 a function of $\MM$ for 
$k\ge 5$. 
A direct calculation confirms this, and shows that $\MM$ is not simply a function of $\PG$ either.
Let 
$$\V=(.2,.3,.4,.5,.7), \qquad
\V'=(.2,.31,.39,.5,.7), \qquad
\V"=(.19,.31,.4,.5,.7),$$
giving
$\PG(\V)=\PG(\V')$ while
$\MM(\V)\neq\MM(\V')$.
On the other hand,
$\MM(\V)=\MM(\V")$ while
$\PG(\V)\neq\PG(\V")$.


\end{document}